\documentclass[10, twocolumn]{IEEEtran}

\usepackage{amsmath, amsfonts, amsthm,  graphicx}
\usepackage{algorithmic, algorithm}
\usepackage{cite}

\newcommand{\Ex}{\mathbb{E}}

\newtheorem{theorem}{Theorem}
\newtheorem{prop}{\bf Proposition}
\newtheorem{lemma}{Lemma}

\newtheorem*{THEO}{Theorem~\ref{thm:synch}'}

\newcounter{definition}
\newenvironment{definition}[1][]{\refstepcounter{definition}\par\medskip\noindent%
   \textbf{Definition~\thedefinition. #1}}{\newline}

\newcounter{example}
\newenvironment{example}[1][]{\refstepcounter{example}\par\medskip\noindent%
   \textbf{Example~\theexample. }}{\medskip}

\newcounter{construction}
\newenvironment{construction}[1][]{\refstepcounter{construction}\par\medskip\noindent%
   \textbf{Construction~\theconstruction. }}{\medskip}

\newcommand{\set}[1]{\mathcal{#1}}
\newcommand{\mat}[1]{\mathbf{#1}}
\renewcommand{\vec}[1]{\boldsymbol{#1}}

\long\def\symbolfootnote[#1]#2{\begingroup
\def\thefootnote{\fnsymbol{footnote}}\footnote[#1]{#2}\endgroup}

\title{Construction and Applications of CRT Sequences}

\author{Kenneth W. Shum, \IEEEmembership{Member, IEEE,} and Wing Shing Wong, \IEEEmembership{Fellow, IEEE}
\thanks{Kenneth W. Shum and Wing Shing Wong are with the Dept. of Information Engineering, %
The Chinese University of Hong Kong, Shatin, Hong Kong. Emails: wkshum@inc.cuhk.edu.hk, wswong@ie.cuhk.edu.hk.}
}

\begin{document}

\maketitle

\begin{abstract}
Protocol sequences are used for channel access in the collision channel without feedback. Each user accesses the channel according to a deterministic zero-one pattern, called the protocol sequence. In order to minimize fluctuation of throughput due to delay offsets, we want to construct protocol sequences whose pairwise Hamming cross-correlation is as close to a constant as possible. In this paper, we present a  construction of protocol sequences which is based on the bijective mapping between one-dimensional sequence and two-dimensional array by the Chinese Remainder Theorem (CRT). In the application to the collision channel without feedback, a worst-case lower bound on system throughput is derived.
\end{abstract}

{\it Tags:} Protocol sequences, collision channel without feedback, cyclically permutable constant-weight codes, optical orthogonal codes.

\section{Introduction}

\subsection{Background and Motivation}
\symbolfootnote[0]{This work was partially supported by a grant from the Research Grants Council of the Hong Kong Special Administrative Region under Project 417909.}
\symbolfootnote[0]{The result in this paper was partially presented in IEEE Int. Symp. on Inform. Theory, Austin, 2010. This paper appears in IEEE Trans. on Inform. Theory, Nov, 2010.}
Randomness is commonly used in the design of multiple-access schemes. For example, in slotted ALOHA, each user transmits a packet with probability $p$ independently. Implementations of such random access schemes in practice usually substitute random variables by pseudo-random numbers. However, high-quality pseudo-random number generation may be too complicated for applications, such as wireless sensor networks, where computing power is limited. The objective of this paper is to construct binary pseudo-random sequences, called protocol sequences, that are tailored to the quality-of-service requirements of interest, such as high throughput and bounded delay.

Protocol sequences are used in multiple-access in the collision channel without feedback~\cite{Massey85}. In this paper, we consider a time-slotted system, consisting of a number of transmitters and one receiver. A user sends a packet within the boundaries of a time slot.  If exactly one user transmits in a time slot, the received packet is received successfully. If two or more users transmit in the same time slot, a collision is incurred, and the received packet is assumed unrecoverable. If no user transmits in a time slot, that time slot is idle. Since there is no feedback from the receiver, collision resolution algorithm such as the stack algorithm is not possible. Each user repeats his assigned binary protocol sequence periodically, and transmits a packet if and only if the value of the protocol sequence at that time slot equals one. We note that the transmission schedule is independent of the data being sent and there is no cooperation among the users.

The capacity of the collision channel without feedback is characterized by Massey and Mathys in~\cite{Massey85}. It is shown that the zero-error sum-capacity is $e^{-1}$.  They use protocol sequences having the special property that the Hamming cross-correlation is independent of relative delay offsets. In fact, the Hamming cross-correlation they considered is a generalized notion of Hamming cross-correlation, which is defined for all nonempty subsets of users, not just for pairs of users. Protocol sequences with this property are called {\em shift-invariant} sequences~\cite{SCSW09}.   Shift-invariant protocol sequences have the advantage that there is no fluctuation in throughput no matter what the delay offsets are, and hence have the largest worst-case system throughput. Constructions of shift-invariant protocol sequences are considered in~\cite{Massey85, Rocha00, SCSW09}. Nevertheless, shift-invariant protocol sequences have the drawback that the period grows exponentially in the number of users~\cite{SCSW09}. Even if we relax this requirement and consider protocol sequence sets with only pairwise Hamming cross-correlation being constant, it is shown in~\cite{ZSW} that the period grows exponentially in the number of users as well. Long period length has the disadvantage that individual or system throughput is invariant only if it is averaged over a long period. Individual users may suffer short-time starvation. In order to achieve short period length, we must seek for protocol sequences with some small variance in Hamming cross-correlation allowed.

After the seminal work of~\cite{Massey85}, more constructions of protocol sequences are given in~\cite{Nguyen92, GyorfiVajda93,  MZKZ95, Bitan95, chung89}, sometime under the name of {\em cyclically permutable constant-weight codes}, or {\em optical orthogonal codes}. The main difference between these works in the literature and our construction  is that, the protocol sequences in these papers are required to have small Hamming cross-correlation and auto-correlation. In our construction, Hamming auto-correlation may be very large.

Another class of protocol sequences, called {\em wobbling sequences}~\cite{Wong07}, has period equal to $M^4$, where $M$ is the number of users, with worst-case system throughput provably larger than a positive constant that is approximately equal to 0.25 when $M$ is large. The result in this paper improves upon the wobbling sequences by constructing protocol sequences of order $O(M^2)$ or $O(M^3)$, depending on the models of user activity, while the guarantee of the worst-case system throughput remains the same.

The construction proposed in this paper is based on the Chinese remainder theorem (CRT), and thereby the constructed sequences are called {\em CRT sequences}.
We remark that the use of CRT in the construction of protocol sequences is not new. The constructions of optical orthogonal codes in~\cite{Nguyen92, GyorfiVajda93, MZKZ95} also employ CRT. However, the specific construction proposed here is novel. Although the analysis of CRT sequences is quite involved, the generation of such sequences requires no more than computing linear operations in modular arithmetics

The proposed CRT sequences have two special features, {\em user-identification} and {\em frame-synchronization} capability.
The sender of each successfully received packet can be identified by looking at the channel activity only, without looking into the packet contents. Also, we can determine the start time of a protocol sequence, again, based on the channel activity only. One potential application based on this property is tracking targets by detecting energy pulses coming from multiple distributed sources, for example, as in a multistatic radar system or in a ultrasound based sensor network. While it may be difficult to identify the transmitting source of a single pulse, if each source employs a CRT sequence and sends multiple pulses according to the corresponding protocol sequence, then the identifiability
property ensures that a detector can identify the transmitting sources of
the pulses and make use of  the information for more accurate tracking of the targets.

\subsection{Main Results}

The main construction is given in Section~\ref{sec:CRT}. One of the key ingredients in the construction is the one-to-one correspondence between one-dimensional sequence and two-dimensional arrays via CRT. The correlation property of the constructed sequences is  analyzed in Section~\ref{sec:cross-correlation}.

Applications to the collision channel without feedback are addressed in Section~\ref{sec:throughput} with two different user activity models~\cite{GG07}. In the first model, an infinite backlog of data is assumed and the users are active throughout the transmission. We show in Theorem~\ref{thm:throughput1} that under this user activity model, we can achieve a throughput of 0.25 with $O(M^{2+\epsilon})$ sequence period, where $M$ is the number of users.

In the second model, users become active only if they have data to send, and remain idle otherwise. If a user becomes active, it is required that the user remains active for at least one period of the protocol sequence assigned. We show in the second part of Section~\ref{sec:throughput} that if the number of active users is no more than one half of the number of potential users, and if the protocol sequences are sufficiently long, namely $O(M^3)$,  the receiver can detect the set of active users and determine their starting time correctly, even without any packet header.

To facilitate comparison between different protocol sequences, we introduce in Section~\ref{sec:uniform} a notion called $\epsilon$-uniformity, which measures the variation of Hamming cross-correlation. Other applications of CRT are given in Section~\ref{sec:discussions}.

\subsection{Definitions and Notations}

Let $\mathbb{Z}_n$ be the ring of residues mod $n$ for a positive integer~$n$.  We will reserve the letter $L$ for sequence length.

\begin{definition}
The components in a sequence of length $L$ are indexed from 0 to $L-1$. The time indices $\{0,1,\ldots, L-1\}$ is identified with~$\mathbb{Z}_L$. The {\em Hamming weight} of a binary sequence $a(t)$ of length $L$ is denoted by $w_a$. The {\em duty factor}~\cite{Massey85} of $a(t)$ is the Hamming weight divided by the length,
\[
 f_a := \frac{1}{L} \sum_{t=0}^{L-1} a(t).
\]
For two binary sequences $a(t)$ and $b(t)$ of length $L$, their {\em Hamming correlation function} is defined as
\[
 H_{ab}(\tau) := \sum_{t=0}^{L-1} a(t) b({t-\tau}),
\]
where $\tau$ is the delay offset, and  $t - \tau$ is the difference in modulo-$L$ arithmetic. When $a(t)=b(t)$,  $H_{aa}(\tau)$ is called the {\em Hamming auto-correlation} of $a(t)$. When $a(t)$ and $b(t)$ are two different sequences,  $H_{ab}(\tau)$ is called the {\em Hamming cross-correlation} of $a(t)$ and $b(t)$.
\end{definition}

When the number of ones in a zero-one sequence is small in compare to the length, the sequence can be compactly represented by specifying the locations of ones.

\begin{definition}
Given a sequence $a(t)$ of length $L$, let the {\em characteristic set} of $a(t)$, denoted by $\set{I}_a$, be the subset of $\mathbb{Z}_L$ such that $t\in \set{I}_a$ if and only if $a(t)=1$.
\end{definition}
Shifting a sequence cyclically by $\tau$ is equivalent to translating its characteristic set by~$\tau$, with addition performed modulo~$L$.
Given a subset $\set{I}$ in $\mathbb{Z}_L$, and $\tau\in \mathbb{Z}_L$, we denote the translation of $\set{I}$ by $\tau$ as
\[ \set{I} + \tau := \{ x +\tau \in \mathbb{Z}_L :\, x\in \set{I} \}.
\]
Expressed in terms of the characteristic set, the Hamming correlation of sequences $a(t)$ and $b(t)$ equals
\[
 H_{ab}(\tau) = | \set{I}_a \cap (\set{I}_b + \tau) |
\]
for all $\tau = 0,1,\ldots, L-1$, where $|\set{S}|$ denote the size of a set~$\set{S}$.

\section{The CRT Construction} \label{sec:CRT}

\subsection{The CRT correspondence}

We shall construct sequences with length $L = pq$, where $p$ and $q$ are relatively prime integers. In subsequent discussions, we  take $p$ to be a prime number and $q$ an integer not divisible by~$p$.


Define a mapping from $\mathbb{Z}_{pq}$ to the direct sum
\[
G_{p,q} := \mathbb{Z}_p \oplus \mathbb{Z}_q
\]
by
\[
 \Phi_{p,q}(x) := (x \bmod p, x \bmod q).
\]
By Chinese remainder theorem~\cite[p.34]{IrelandRosen}, $\Phi_{p,q}$ is a bijective map.
We will call $\Phi_{p,q}$ the {\em CRT correspondence}. When the values of $p$ and $q$ are clear from the context, we  write $\Phi(x)$ instead of $\Phi_{p,q}(x)$.

It can be easily checked that the CRT correspondence is a linear map, meaning that
\[\Phi_{p,q}(x+x') = \Phi_{p,q}(x) + \Phi_{p,q}(x').\]
Here, the addition on the left hand side is the addition in $\mathbb{Z}_{pq}$, and the addition on the right hand side is the addition in $G_{p,q}$.

A sequence $s(t)$ of length $L$ is associated with a $p\times q$ array $\mat{S}(t_1,t_2)$, where $t_1$ and $t_2$ range from 0 to $p-1$ and  0 to $q-1$ respectively, via the relation
\[ \mat{S}(t \bmod p, t \bmod q) = s(t). \]
The corresponding characteristic set of $s(t)$ in $\mathbb{Z}_L$ is mapped to a subset of $G_{p,q}$ via $\Phi_{p,q}$ as well.

Under the CRT correspondence, a one-dimensional cyclic shift of $s(t)$ by one time unit is equivalent to a column-wise shift followed by a row-wise shift. One-dimensional correlation properties can be translated to the two-dimensional ones.

For illustration, consider the time indices $(0,1,2,\ldots, 14)$ as an integer sequence of length $L=15$,  $p=3$ and $q=5$.  By the bijection $\Phi_{3,5}$, this integer sequence is mapped to
\[
\begin{bmatrix}
0  & 6 & 12 & 3  & 9\\
10 & 1 & 7 & 13  & 4 \\
5  & 11 & 2 & 8  & 14
\end{bmatrix}.
\]
When $(0,1,\ldots, 14)$ is cyclically shifted to $(14,0,1,2,\ldots, 13)$, the corresponding array is
\[
\begin{bmatrix}
14  & 5 & 11 & 2  & 8\\
9 & 0 & 6 & 12  & 3 \\
4  & 10 & 1 & 7  & 13
\end{bmatrix}.
\]
Note that the second array can be obtained from the first one by cyclically shift downward by one row and then to the right one column.

For $\vec{\tau} = (\tau_1, \tau_2) \in G_{p,q}$,
define the Hamming correlation between two 2-dimensional arrays $\mat{A}$ and $\mat{B}$ by
\[
 H_{\mat{AB}}(\vec{\tau}) := \sum_{\vec{t}} \mat{A}(\vec{t}) \mat{B}(\vec{t} - \vec{\tau}),
\]
with the subtraction calculated in $G_{p,q}$, and $\vec{t}$ running over all elements in~$G_{p,q}$. It is easy to check that this definition of Hamming correlation is compatible with the 1-dimensional analog, i.e.,
\[
  H_{\mat{AB}}(\tau_1, \tau_2) = H_{ab}(\tau)
\]
with $\tau_1 \equiv \tau \bmod p$ and $\tau_2 \equiv \tau \bmod q$.

The Hamming correlation between two 2-dimensional arrays $\mat{A}$ and $\mat{B}$ can be expressed in terms of their characteristic set $\set{I}_\mat{A}$ and $\set{I}_\mat{B}$ as
\[
 H_{\mat{AB}}(\vec{\tau}) = | \set{I}_{\mat{A}} \cap (\set{I}_{\mat{B}} + \vec{\tau}) |,
\]
where $\set{I}_{\mat{A}}$ denotes the characteristic set of $\mat{A}$,
\[
\set{I}_{\mat{A}} := \{ (i,j) \in \mathbb{Z}_p \oplus \mathbb{Z}_q:\,
\mat{A}(i,j) = 1\},
\]
and the addition in $\set{I}_{\mat{B}} + \vec{\tau}$ is performed in $\mathbb{Z}_p \oplus \mathbb{Z}_q$.

\subsection{The CRT sequences}

We will construct sequences by specifying characteristic sets in~$G_{p,q}$. Rows and columns of matrices and arrays  will be indexed by $\{0,1,\ldots, p-1\}$ and $\{0,1,\ldots, q-1\}$ respectively.

\begin{definition}  Let $p$ be prime and $q$ be an integer not divisible by~$p$. For $g  \in \mathbb{Z}_p$, we define
\begin{equation}
\set{I}_{g,p,q}:=  \{ (g,1)t \in G_{p,q}:\, 0\leq t < q\}.  \label{eq:CRT}
\end{equation}
The notation $(g,1)t$ simply means the sum of $t$ copies of $(g,1)$ in $G_{p,q}$,
\[
(g,1)t := \underbrace{(g,1) + (g,1) + \ldots + (g,1)}_t.
\]
We say that $\set{I}_{g,p,q}$ is generated by~$g$, and $g$ is the {\em generator} of $\set{I}_{g,p,q}$. If $p$ and $q$ are clear from the context, we write~$\set{I}_g$ instead of $\set{I}_{g,p,q}$. \label{def:array}
\end{definition}

The elements in $\set{I}_{g,p,q}$ form an arithmetic progression in $G_{p,q}$ with common difference $(g,1)$. We can re-write $\set{I}_{g,p,q}$ in the following form
\[
 \set{I}_{g,p,q} = \{ (gt, t) \in G_{p,q}:\, 0 \leq t < q\},
\]
with the product $gt$ in the first component reduced mod $p$.

{\em Remarks:}
For $g=0,1,\ldots, p-1$, the array with $\set{I}_{g,p,q}$ as the characteristic set contains exactly one ``1'' in each column.  For $g\neq 0$, each block of $p$ consecutive columns form a permutation matrix. (Recall that a permutation matrix is a square zero-one matrix with exactly one ``1'' in each row and each column.)

\begin{definition} (CRT sequences)  For $g=0,1,\ldots, p-1$, define the {\em CRT sequence generated by~$g$}, denoted by $s_{g,p,q}(t)$, be the binary sequence of length $L$ obtained by setting
\[
 s_{g,p,q}(t) = \begin{cases}
 1 & \text{ if } \Phi_{p,q}(t) \in\set{I}_{g,p,q} \\
 0 & \text{ otherwise}.
 \end{cases}
\]
We will write $s_g(t)$ if the values of $p$ and $q$ are understood. \label{def:CRT}
\end{definition}



\begin{example} $p=3$ and $q = 5$.
The three characteristic sets are: {\small
\begin{align*}
 \set{I}_{0,3,5} &= \{ (0,0), (0,1), (0,2), (0,3), (0,4)\}, \\
 \set{I}_{1,3,5} &= \{ (0,0), (1,1), (2,2), (0,3), (1,4)\}, \\
 \set{I}_{2,3,5} &= \{ (0,0), (2,1), (1,2), (0,3), (2,4)\}.
\end{align*} }
The three arrays are shown in Fig.~\ref{fig:exampleA1}. The top left corner in each array is the $(0,0)$-entry.
The generated CRT sequences are listed as follows
\begin{align*}
s_{0}(t):&\ 10010\,01001\,00100\\
s_{1}(t):&\ 11111\,00000\,00000\\
s_{2}(t):&\ 10010\,00100\,01001.
\end{align*}
\label{ex:A}
\end{example}

\begin{figure}
\[ \set{I}_{0,3,5}:\ { \small
\begin{array}{|ccccc|} \hline
1&1&1&1&1 \\
0&0&0&0&0 \\
0&0&0&0&0 \\ \hline
\end{array} }
\]

\[ \set{I}_{1,3,5}:\ {\small
\begin{array}{|ccccc|} \hline
1&0&0&1&0 \\
0&1&0&0&1 \\
0&0&1&0&0 \\ \hline
\end{array} }
\]

\[ \set{I}_{2,3,5}:\ {\small
\begin{array}{|ccccc|} \hline
1&0&0&1&0 \\
0&0&1&0&0 \\
0&1&0&0&1 \\ \hline
\end{array} }
\]
\caption{The three arrays associated with characteristic sets $\set{I}_{0,3,5}$, $\set{I}_{1,3,5}$, $\set{I}_{2,3,5}$.}
\label{fig:exampleA1}
\end{figure}

CRT sequences satisfy the following properties.

\begin{theorem}\

\begin{enumerate}
\item The Hamming weight of each CRT sequence is~$q$.

\item $s_0(t+p) = s_0(t)$, i.e., the least period of sequence $s_0(t)$ is~$p$.

\item For each residue $i \bmod q$, there is exactly one ``1'' at $t=i, i+q, i+2q, \ldots, i+(p-1)q$ in $s_g(t)$.
\end{enumerate}
\label{thm:I}
\end{theorem}

\begin{proof}
As the characteristic set $\set{I}_{g,p,q}$  contains $q$ elements and the CRT correspondence $\Phi$ preserves Hamming weight, the first statement in the proposition follows immediately.

For $g=0$, the elements in characteristic set $\set{I}_{0,p,q}$ have the first coordinate identically equal to zero. For $j=0,1,\ldots, q-1$, The pre-image of $(0,j)$ under $\Phi$ is a multiple of $p$.
We have $s_0(t) = 1$ whenever $t$ is a multiple of~$p$.

For the last statement in the proposition, consider the image of $i+kq$, for $k=0,1,\ldots, p-1$, under the mapping $\Phi$,
\[
 \Phi(i+kq) = (i+kq \bmod p,  i \bmod q).
\]
The second coordinate is constant. If $k$ goes through $0$ to $p-1$,  $\Phi(i+kq)$ will go through the $i$th column in the array with characteristic set $\set{I}_{g,p,q}$. Because there is exactly one ``1'' in each column, there is exactly one ``1'' in positions $t=i, i+q, \ldots, i+(p-1)q$.
\end{proof}


\section{Correlation Properties of CRT Sequences} \label{sec:cross-correlation}

We continue to use the notation that $p$ is a prime number, $q$ is an integer not divisible by $p$.
Let $L=pq$ denotes the sequence length. In this section, we determine the Hamming correlation of the CRT sequences.  In this section, we will use the notation $\bar{x}$ for the remainder of $x$ after division by~$p$. We distinguish the two additions mod $p$ and mod $q$ by $\oplus_p$ and $\oplus_q$, respectively.

For $g=0,1,\ldots, p-1$, the $p\times q$ array corresponding to sequence $s_{g}(t)$ is denoted by $\mat{A}_g$. Recall that the characteristic set of $\mat{A}_g$ is given in~\eqref{eq:CRT}, and $\mat{A}_g(i,j) = 1$ if and only if $i \equiv \bar{j}g \bmod p$, for $i=0,\ldots,p-1$ and $j=0,\dots, q-1$. For notational convenience, we let
\begin{equation}
  H_{gh}(\tau_1, \tau_2) := H_{\mat{A}_g \mat{A}_h}(\tau_1, \tau_2).
  \label{def:H2D}
\end{equation}

\subsection{Hamming Cross-correlation}

Let $g$ and $h$ be two distinct elements in $\mathbb{Z}_p$.
As argued in the previous section, the  cross-correlation of CRT sequences $s_g(t)$ and $s_h(t)$ is equivalent  to the cross-correlation of the associated $p\times q$ arrays $\mat{A}_g$ and $\mat{A}_h$, namely by counting the number of elements in
common to $\set{I}_g$ and $\set{I}_h+(\tau_1, \tau_2)$.
By definition, the translation $\set{I}_h+(\tau_1, \tau_2)$ of $\set{I}_h$ is
\begin{equation}
 \{ ((\bar{j}h)\oplus_p \tau_1, j \oplus_q \tau_2):\, j=0,1,\ldots, q-1\}.
 \label{eq:cross-correlation}
\end{equation}
 By a change of variable, \eqref{eq:cross-correlation} can be written as
\[
 \{ ((( \overline{j \ominus_q \tau_2})h) \oplus_p \tau_1, j):\, j=0,1,\ldots, q-1\}.
\]
where $\ominus_q$ denotes subtraction mod $q$. In the calculation of $( \overline{j \ominus_q \tau_2})h$,  $j \ominus_q \tau_2$ must be first reduced to an integer between 0 and $q-1$ before we reduce modulo~$p$. The intersection of $\set{I}_g$ and $\set{I}_h + (\tau_1, \tau_2)$ equals to the number of solutions to
\begin{equation}
 \bar{x}g  \equiv ((\overline{x \ominus_q \tau_2})h) \oplus_p  \tau_1 \bmod p. \label{eq:Ham}
\end{equation}
for $x = 0,1,\ldots, q-1$.

For nonzero $h$, we can divide both sides of \eqref{eq:Ham} by $h$ and re-write it as
\[
  \bar{x} (h^{-1} g) \equiv (\overline{x \ominus_q \tau_2}) \oplus_p (h^{-1} \tau_1) \bmod p.
\]
For each fixed $\tau_2$, as $\tau_1$ runs through $\mathbb{Z}_p$, $h^{-1} \tau_1$ also runs through $\mathbb{Z}_p$. Therefore, the distribution of Hamming cross-correlation between $s_g(t)$ and $s_h(t)$ is the same as the distribution of Hamming cross-correlation between $s_{g/h}(t)$ and $s_1(t)$. From this observation, it suffices to take $h=1$ without any loss of generality. We will consider the Hamming cross-correlation  between $s_g(t)$ and $s_1(t)$, for $g=0$ and $g=2,3,\ldots, p-1$.

The following simple lemma is used repeatedly in the derivation of Hamming cross-correlation properties.

\begin{lemma} For each $b \in\mathbb{Z}_p$, the number of solutions to
\[
\bar{x} \equiv b \bmod p
\]
for $x$ going through $d$ consecutive integers $c, c+1, \ldots c+d-1$, equals
\[
  \begin{cases}
   d/p  & \text{if $p$ divides $d$}, \\
  \lfloor d/p \rfloor + \delta & \text{otherwise,}
   \end{cases}
\]
where $\delta$ equals either 0 or~1. \label{lemma:simple}
\end{lemma}

\begin{proof}
Suppose that $d$ is divisible by~$p$. If we reduce the integers $c$, $c+1, \ldots, c+d-1$ mod $p$, we have each element in $\mathbb{Z}_p$ repeated $d/p$ times. Hence, for each $b \in \mathbb{Z}_p$, there are exactly $d/p$ integers in $\{c, c+1, \ldots, c+d-1\}$ whose residue mod $p$ equal~$b$

Now suppose that $d$ is not divisible by~$p$, we divide the $d$ consecutive integers into two parts. Among the first $\lfloor d/p \rfloor p$ integers, for each $b \in \mathbb{Z}_p$, exactly $\lfloor d/p \rfloor$ equals $b$ mod $p$. The residues of the remaining $d - \lfloor d/p \rfloor p$ integers are distinct. The number of integers in $\{c, c+1, \ldots, c+d-1\}$ whose residues equal $b$ is either $\lfloor d/p \rfloor$ or $\lfloor d/p \rfloor + 1$.
\end{proof}

\medskip

We obtain the following theorem immediately from Lemma~\ref{lemma:simple}.

\begin{theorem}
The Hamming cross-correlation of $s_1(t)$ and $s_0(t)$ is equal to either $\lfloor q/p \rfloor$ or $\lfloor q/p \rfloor +1$. \label{thm:cross-correlation0}
\end{theorem}

\begin{proof}
If we put $g=1$ and $h=0$ in~\eqref{eq:Ham}, we get
\[
  \bar{x} \equiv \tau_1 \bmod p.
\]
The number of integers in $\{0,1,\ldots, q-1\}$ that equal $\tau_1 \bmod p$ is either  $\lfloor q/p \rfloor$ or $\lfloor q/p \rfloor +1$
by Lemma~\ref{lemma:simple}.
\end{proof}

From now on, we  assume $q > p$, which is the case of practical interest.

\begin{theorem} Suppose $q>p$. Let $m$ be the quotient of $q$ divided by~$p$, i.e., $m=\lfloor q/p \rfloor$, and let $g\in\mathbb{Z}_p$, $0 \neq g\neq 1$. Let $\bar{q}$ be the residue of $q$ mod $p$, and
\begin{equation}b_g \equiv (g-1)^{-1} \bar{q} \bmod p.
\label{eq:bg}
\end{equation}
The Hamming cross-correlation between $s_g(t)$ and $s_1(t)$ is bounded between
\[
\begin{cases}
m-1 \text{ and } m+1 & \text{if } 0 < b_g < p-\bar{q}, \text{ or }\\
m \text{ and } m+2 & \text{if } p-\bar{q} < b_g < p.
\end{cases}
\]
\label{thm:cross-correlation1}
\end{theorem}

The proof of Theorem~\ref{thm:cross-correlation1} is in Appendix~\ref{app:proof_of_cross1}. We remark that in Theorem~\ref{thm:cross-correlation1}, $b_g$ is neither 0 nor $p-\bar{q}$.
The value of  $b_g$ is nonzero mod $p$ because both $g-1$ and $\bar{q}$ are nonzero. On the other hand, $b_g$ is equal to $p-\bar{q}$ only if $g=0$, which is excluded by assumption.

\medskip

With more careful book-keeping, we can determine exactly the frequencies of occurrence of Hamming cross-correlation.

\begin{definition}
Let the distribution of $H_{g1}(\tau_1, \tau_2)$ be
\begin{equation}
 N_{g}(j) := |\{ (\tau_1, \tau_2) \in G_{pq}:\, H_{g1}(\tau_1,\tau_2)=j\}|,
 \label{eq:N}
\end{equation}
for $j=0,1,\ldots, q$.
\end{definition}

\begin{theorem} With notation as in Theorem~\ref{thm:cross-correlation1}, we have

\begin{enumerate}
\item $N_0(m) = (p-\bar{q})q$, $N_0(m+1) = \bar{q}q$.

\item If $0< b_g <p-\bar{q}$, then
\begin{align}
N_g(m-1) &= \eta \label{eq:dist1}\\
N_g(m) &= q (p-\bar{q})- 2\eta \label{eq:dist2}\\
N_g(m+1) &= q \bar{q}+ \eta \label{eq:dist3},
\end{align}
where
\[
 \eta := m b_g (p-b_g-\bar{q}).
\]

\item If $p-\bar{q} < b_g <p$, then
\begin{align}
N_g(m) &= q (p-\bar{q})+ \theta \label{eq:dist4} \\
N_g(m+1) &= q \bar{q}- 2\theta , \label{eq:dist5}\\
N_g(m+2) &= \theta \label{eq:dist6}
\end{align}
where
\[
 \theta :=(m+1)(p-b_g)(\bar{q} +b_g - p).
\]
\end{enumerate}
\label{thm:frequency}
\end{theorem}

The proof is relegated to Appendix~\ref{app:proof_of_cross2}.

\begin{example}
For the CRT sequences in Example~\ref{ex:A} with parameters $p=3$ and $q=5$,
we tabulate the distribution of Hamming cross-correlation between $s_g(t)$, and $s_1(t)$, for $g \neq 1$,  as follows.
\begin{center}
\begin{tabular}{|c|c||c|c|c|} \hline
$g$ & $b_g$&  $N_g(1)$ & $N_g(2)$ & $N_g(3)$ \\ \hline \hline
0 & 1 &  5 & 10 & 0 \\ \hline
2 & 2 &  7 & 6 & 2\\ \hline
\end{tabular}
\end{center}
We note that if $g=2$, then $b_g$ is 2 mod 3. By  Theorem~\ref{thm:cross-correlation1}, $H_{21}(\tau)$ equals 1, 2, or 3. The Hamming cross-correlation is distributed according to the third part of Theorem~\ref{thm:frequency}.
\end{example}

When $q \equiv \pm 1 \bmod p$, the Hamming cross-correlations have three distinct values, for all pairs of distinct CRT sequences chosen from $s_1(t)$, $s_2(t), \ldots, s_{p-1}(t)$.

\begin{theorem} Suppose $q>p$.

\begin{enumerate}

\item Let $q$ be of the form $mp+1$. For $g=2,3,\ldots, p-1$, $H_{g1}(\tau)$ is between $m-1$ and $m+1$.

\item Let $q$ be of the form $mp+(p-1)$. For $g = 2,3,\ldots, p-1$, $H_{g1}(\tau)$ is between $m$ and $m+2$.
\end{enumerate}
\label{thm:cross-correlation2}
\end{theorem}

\begin{proof}
For the first part of the theorem, we have $\bar{q}$ equal to 1 mod~$p$.  So
\[
b_g \equiv (g-1)^{-1} \bar{q} \equiv (g-1)^{-1} \bmod p.
\]
When $g$ runs from $2$ to $p-1$, the value of $g-1$ runs over all elements in $\mathbb{Z}_p$ except  0 and $p-1$. For odd $p$, the multiplicative inverse of $p-1$ is itself. Hence the range of $(g-1)^{-1} \bmod p$ is $\mathbb{Z}_p\setminus \{0,-1\}$. We thus obtain
$0< b_g < p-1$. The result now follows from Theorem~\ref{thm:cross-correlation1}.

The second part can be proved similarly.
\end{proof}

\subsection{Hamming Auto-correlation}
The distribution of the Hamming auto-correlation can also be determined explicitly.

\begin{theorem}
\[
H_{00}(\tau) =
\begin{cases}
q & \text{if $\tau$ is a multiple of $p$}, \\
0 & \text{otherwise}.
\end{cases}
\]

For $g=1,2,\ldots, p-1$, and $k=0,1,\ldots, q-1$,
\[
H_{gg}(\tau) =
\begin{cases}
q-k & \text{if } \Phi(\tau) = \pm  (g,1)k, \\
0 & \text{otherwise}.
\end{cases}
\]
\label{thm:auto-correlation}
\end{theorem}

\begin{proof}
The first statement follows by the second part of Theorem~\ref{thm:I} that $s_0(t)$  has least period~$p$.

For the second statement, we note that $(g,1)t$ runs through the whole group $G_{p,q}$ as $t$ runs from 0 to $pq-1$. This follows from CRT, because for any $x$ and $y$,
\begin{align*}
gt \equiv x &\mod p \\
t  \equiv y &\mod q
\end{align*}
has one and only one solution mod $pq$. Via a bijective mapping
$t \mapsto (g,1)t$, the set $\set{I}_g$ can be regarded as an arithmetic progression $\{0,1,2,\ldots, q-1\}$ in $\mathbb{Z}_{pq}$.
The intersection $\set{I}_g \cap (\set{I}_g + (i,j))$ is a subset of $\set{I}_g$ that is also an arithmetic progression. The intersection $\set{I}_g \cap (\set{I}_g + (i,j))$ is nonempty if and only $\pm (i,j) \in \set{I}_g$. If $|\set{I}_g \cap (\set{I}_g + (i,j))|$  is nonzero, then $(i,j)$ is equal to $\pm(g,1)k$ for some $k=0,1,\ldots, q-1$, and
\[
| \set{I}_g \cap (\set{I}_g + k(g, 1))| = q-|k|.
\]
\end{proof}

\section{Application to the Collision Channel without Feedback}
\label{sec:throughput}

Consider a time-slotted collision channel with $K$ transmitters and one receiver as described in the introductory section in this paper, and assume that there is no cooperation among the users, and no feedback from the receiver. We regard a packet as a $Q$-ary alphabet chosen from the alphabet set~$\Omega = \{1,2,\ldots, Q\}$. The channel output at slot $t$ equals
\[
\begin{cases}
0 & \text{if no user transmits at slot $t$}, \\
* & \text{if two or more users transmit at slot $t$}, \\
x & \text{if exactly one user transmits a packet with content $x$}.
\end{cases}
\]

Each user is assigned a deterministic and periodic zero-one sequence, called protocol sequence~\cite{Massey85}.
For $i=1,2,\ldots, K$, the protocol sequence associated with user~$i$ is denoted by $s_i(t)$, which is a periodic binary sequence of period~$L$.
As there is no feedback from the receiver and no cooperation among the users, each user has a relative delay offset $\tau$, which is random but remains fixed throughout the communication session. User $i$ sends a packet at slot $t$ if $s_i(t+\tau) = 1$, and remains silent if $s_i(t+\tau) = 0$.

We consider two different factors in the system model: (i) systems with packet header or without packet header, and (ii) systems with permanent user activity or partial user activity.
For system without packet header, the receiver has to identify the set of active users (user identification), and determine the sender of each successfully received packet (packet identification). If packet header is present, information such as user identity is obtained readily by looking up the header, and there is no user and packet identification problem.

For system with partial user activity, the number of active users $M$ is smaller than the total number of potential users $T$. When a user changes from inactive to active, packets are sent according to the protocol sequence assigned. It is assumed that after all the data are sent, the user must remain inactive for at least $L$ slots before becoming active again. On the other hand, if permanent user activity is assumed, the number of active users $M$ is equal to the total number of users, and each user sends packets periodically according to the protocol sequence.

\subsection{Permanent User Activity with Header}
For systems with a packet header, the users are simply identified by packet headers.
We calculate the achievable throughput when CRT sequences are used. Let the number of users be $M$. We pick $M$ sequences from the CRT sequences generated with parameters $p$ and~$q$. To take advantage of the three-valued property mentioned in Theorem~\ref{thm:cross-correlation2} we pick a $q$ which is $-1$ mod~$p$.

\begin{construction}
Let $p$ be a prime number and $k$ a positive integer. We choose $q= kp-1$.  The CRT construction in Definition~\ref{def:CRT} yields $p$ protocol sequences of length $L=kp^2-p$ and weight $w=kp-1$.  We pick any $M$ sequences from this set of CRT sequences.
\label{construction:b}
\end{construction}

Since the Hamming weight is $kp-1$ and the pairwise Hamming cross-correlation is at most $k+1$ by Theorem~\ref{thm:cross-correlation2}, the total number of successful packets (summed over all $M$ users) is lower bounded by
\begin{equation}
M  [kp-1 - (M-1)(k+1)].  \label{eq:lower_bound}
\end{equation}
By completing square, we can write the above as
\begin{equation}
 (k+1)\left[ -\Big( M - \frac{k(p+1)}{2(k+1)} \Big)^2 + \Big(  \frac{k(p+1)}{2(k+1)} \Big)^2 \right]. \label{eq:lower_bound1}
\end{equation}

Consider~\eqref{eq:lower_bound1} as a function of $M$. The maximum value is obtained when
\begin{equation}
 M^* = \frac{k(p+1)}{2(k+1)} , \label{eq:best_M}
\end{equation}
and the  maximal value is
\begin{equation}
 \frac{(p+1)^2}{4} \cdot \frac{k^2}{k+1}. \label{eq:opt_throughput}
\end{equation}
After dividing the above by the period, we obtain the following lower bound on throughput.

\begin{theorem}
Let
\[
\tilde{M}=  \Big\lfloor \frac{k(p+1)}{2(k+1)} \Big\rfloor
\]
in Construction~\ref{construction:b}.
The system throughput is lower bounded by
\begin{equation}
\frac{1}{4}\cdot \frac{(p+1)^2}{kp^2-p} \cdot\frac{k^2}{k+1} - \frac{k+1}{kp^2-p} \label{eq:throughput}
\end{equation}
\label{thm:throughput1}
\end{theorem}

\begin{proof}
Consider the expression in~\eqref{eq:lower_bound1} as a function of~$M$, and denote it by $f(M)$. The coefficient of $M^2$ is $-(k+1)$. If we take $M = \tilde{M}$, we deviate from the optimal value of $M^*$ by at most~1. Therefore,
\[
 f(M^*) - f(\tilde{M}) = (k+1)(\tilde{M} - M^*)^2 \leq k+1.
\]
If we put $M=  \lfloor \frac{k(p+1)}{2(k+1)} \rfloor$ in~\eqref{eq:lower_bound}, we drop from the optimal value in~\eqref{eq:opt_throughput} by at most $k+1$. Hence the total number of successful packets, divided by the period, is lower bounded by~\eqref{eq:throughput}.
\end{proof}

Theorem~\ref{thm:throughput1} provides a lower bound on the worst-case throughput. The mean system throughput, averaged over all delay offsets, is however much higher than the lower bound.

\begin{example}We consider an example with $M=19$ users, using CRT sequences with $p=37$. The throughput is plotted against sequence length with increasing $k$, while keeping the duty factor fixed at $1/p$. We compare the lower bound in~\eqref{eq:throughput} of worst-case throughput with the average throughput obtained by simulation in Fig.~\ref{fig:simulate}. For each~$k$, 100000 delay offset combinations are randomly generated. Beside the mean throughput, the maximum and minimum throughput obtained among these 100000 delay offset combinations are also plotted.  We can observe that the mean system throughput is about 0.31, in accordance with the theoretical value
\[
M \cdot \frac{1}{p} \Big( 1 - \frac{1}{p} \Big)^{M-1} =  \frac{19}{37}  \Big( 1 - \frac{1}{37} \Big)^{18} = 0.314.
\]
When the length increases, the lower bound approaches $0.25(p+1)^2/p^2 - 1/p^2= 0.263$.
\end{example}

\begin{figure}
\begin{center}
  \includegraphics[width=3.8in]{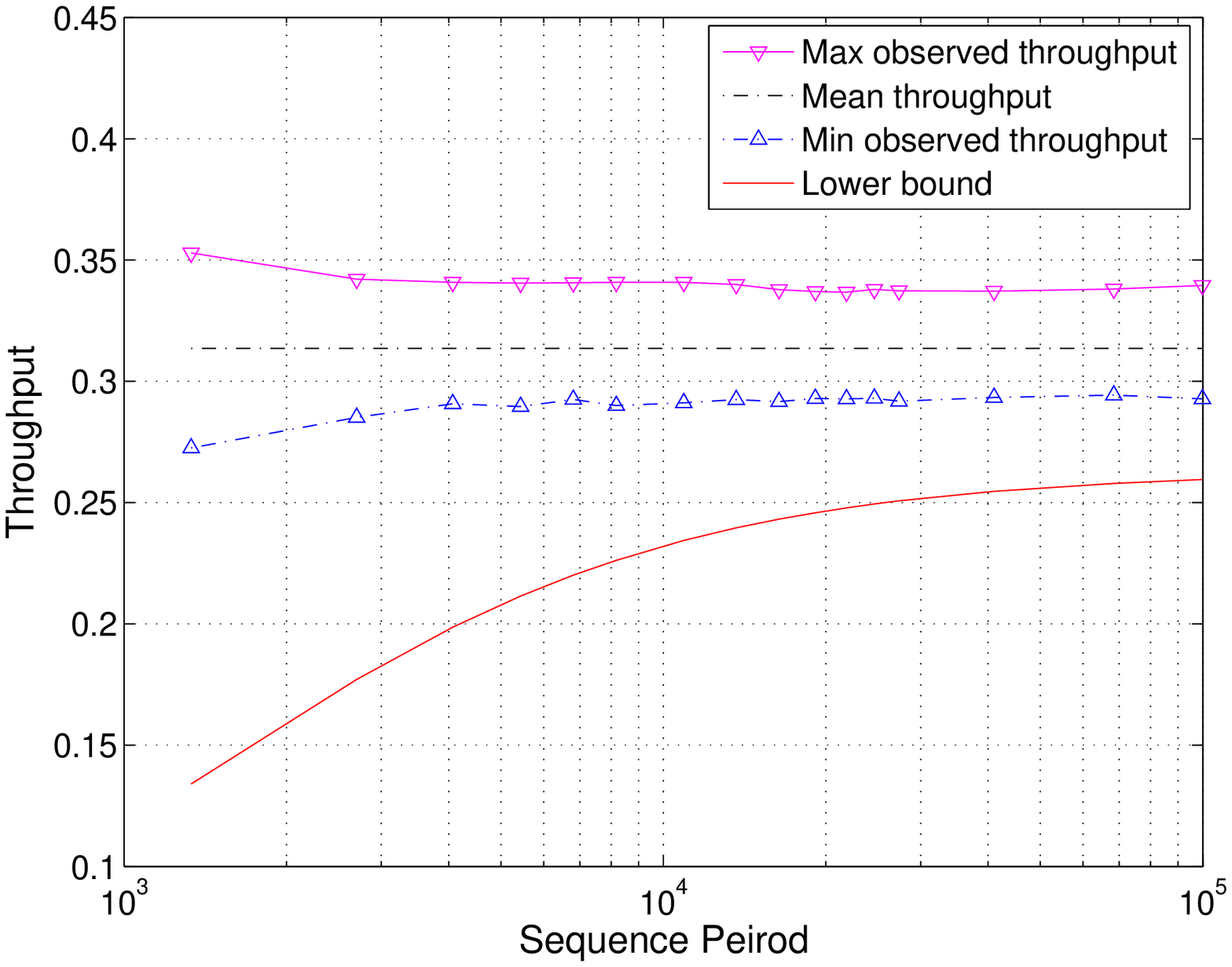}
\end{center}
\caption{System Throughput of CRT Sequences for 19 Users, $p=37$.}
\label{fig:simulate}
\end{figure}

Asymptotically, if we increase $k$ and $p$ in such a way that $k$ increases much slower than  $p$, we obtain a lower bound on the system throughput of $1/4$.

\begin{theorem} For arbitrarily small $\epsilon, \delta > 0$, there exists an infinite class of protocol sequence sets of length $O(M^{2+\epsilon})$, where $M$ is the number of sequences, with system throughput lower bounded by $0.25-\delta$ for all sufficiently large~$M$.  \label{thm:onequater}
\end{theorem}

\begin{proof}
We choose $k=\log(p)$. The lower bound of system throughput in Theorem~\ref{thm:throughput1} tends to $1/4$, with
$$L \sim kp^2 \sim k(2M)^2 \sim 4M^2 \log M.$$ We can find an integer $M_0$ sufficiently large such that $4M_0^2 \log M_0$ is less than $M_0^{2+\epsilon}$. The theorem then holds for all $M \geq M_0$.
\end{proof}

\subsection{Partial User Activity without Header}

In this application, we use a modified version of the CRT correspondence. Let $\gamma$ be the multiplicative inverse of $p$ in $\mathbb{Z}_q$, i.e., $\gamma p \equiv 1 \bmod q$. Since $p$ and $q$ are relatively prime, such inverse exists. Define the mapping $\Phi_{p,q}': \mathbb{Z}_{pq} \rightarrow G_{p,q}$ by
\[
 \Phi_{p,q}'(x) := (x \bmod p, \gamma x \bmod q).
\]
It can also be shown that $\Phi'_{p,q}(x)$ is a group isomorphism between $\mathbb{Z}_{pq}$ and $G_{p,q}$.

Because of its unfavorable Hamming auto-correlation property, the sequence generated by $g=0$ is not used in this application. Given a prime number $p$, the system supports $p-1$ potential users. We label the users from 1 to $p-1$.

\begin{construction}
 Let $p$ be a prime number, and $q$ be an integer relatively prime to $p$. Choose $\gamma$ as described above. For $g=1,\ldots, p-1$, construct the CRT sequence $s_g(t)$ of length $pq$ by setting
\[
 s_g(t) = \begin{cases}
  1 & \text{ if } \Phi'_{p,q}(t) \in \set{I}_{g,p,q} \\
  0 & \text{ otherwise}.
 \end{cases}
\]
 \label{construction:async}
\end{construction}

The sequence generated by $g$ is assigned to user~$g$. We have used $\Phi'_{p,q}(x)$ in the above construction, instead of $\Phi_{p,q}(x)$ as in Construction~\ref{construction:b}. The cross-correlation properties of the resulting modified CRT sequences are exactly the same as in the previous section, because the proof in the previous section is essentially about two-dimensional Hamming cross-correlation. If $\Phi$ is replaced by $\Phi'$, all theorems about Hamming cross-correlation and auto-correlation also hold.

\medskip

{\bf Sequence synchronization}

For systems without a header, we need to identify the sender of a successful packet. Also, as users may come and go, we also need to determine when a user becomes active. We show that these two tasks can be achieved by merely observing the channel activity, that is whether a time slot is idle, containing a collision or a successful transmission, without looking into the packet contents.

\begin{definition}
For each time index $t$, let $c(t)$ be 0 if it is an idle slot, 1 if exactly one user transmits a packet, or * if two or more users transmit. We  call $c(t)$ the {\em channel-activity signal}. We say that $c(t)$ is {\em matched} to $s_i(t)$ at time $t_0$ if
$\forall t = 0,1,\ldots, L-1,$ $s_i(t) = 1 \Rightarrow c(t_0+t) = 1 \text{ or } *$.
\end{definition}

The receiver stores the channel-activity signal in a first-in-first-out queue.

We want to determine (a) the time when a user becomes active, and (b)  the time when an active user change status from active to idle. The receiver keeps track of the active users by maintaining $p-1$ Boolean variables $active(i)$, for $i=1,2,\ldots, p-1$. The value of $active(i)$ is set to $\tt{FALSE}$ if the user is idle, and $\tt{TRUE}$ if the user is active. To determine whether user $i$ becomes active at time $t_0$, the receiver wait until time $t_0+L$. At that time the channel-activity signal $c(t_0)$, $c(t_0+1), \ldots, c(t_0+L-1)$ are available. The receiver then checks whether $c(t)$ is matched to $s_i(t)$ at time~$t_0$. If there is a match, we declare that user $i$ is active and the starting time of user $i$ is stored in variable $start(i)$.
We summarizes the procedure in Algorithm~\ref{alg:active} below.

\begin{algorithm}
\begin{algorithmic}[1]
    \FOR{$i=1,2,\ldots, p-1$}
    \IF{$active(i) = \tt{FALSE}$}
       \IF{$c(t)$ is matched to $s_i(t)$ at $t_0$}
          \STATE $active(i) \leftarrow \tt{TRUE}$
           \STATE $start(i) \leftarrow t_0$
       \ENDIF
     \ELSE
        \IF{$t_0 - start(i)$ is a multiple of $L$ $\mathbf{and}$ $c(t)$ is not matched to $s_i(t)$ at time $t_0$}
             \STATE  $active(i) \leftarrow \tt{FALSE}$
             \STATE  $start(i) \leftarrow 0$
        \ENDIF
     \ENDIF
    \ENDFOR
\end{algorithmic}
\caption{Determining when a user becomes active or inactive at time $t_0$.}
\label{alg:active}
\end{algorithm}

Under some conditions on the number of active users and period length, we can show that the above algorithm is able to identify the starting time of each user.

\begin{theorem} Let $p$ be a prime number, and $q > 2p^2$ be an integer relatively prime to $p$.
Construct $p-1$ CRT sequences by Construction~\ref{construction:async} and assign them to users 1 to $p-1$. Suppose that at most $(p+1)/2$ users are active at the same time. Then Algorithm~\ref{alg:active} can successfully identify the active users and determine their starting time.
\label{thm:synch}
\end{theorem}

The proof of Theorem~\ref{thm:synch} is given in Appendix~\ref{app:proof_of_synch}.

\begin{example}
Consider the CRT sequences generated with parameters $p=7$, $q=8$ and $\gamma=7$. The period is equal to 56. Suppose that users 1, 2, 3, 4 and 6 are active. The relative delay offsets of users 1 and 2 are 0, and  the relative delay offsets of users 3, 4 and 6 are~1. The CRT sequences $s_1(t)$ to $s_4(t)$ and $s_6(t)$ and the induced channel-activity signal $c(t)$ are shown in Fig.~\ref{fig:synch}
\begin{figure*}
\begin{align*}
s_1(t):&\ \underbrace{10001000\ \ 00000001\ \ 00010000\ \ 00000010\ \ 00100000\ \ 00000100\ \ 01000000}\ \ 10 \\
s_2(t):&\ \underbrace{10010000\ \ 00000100\ \ 00000001\ \ 00100000\ \ 00001000\ \ 00000010\ \ 01000000}\ \ 10 \\
s_3(t-1):&\ 0\underbrace{1000001\ \ 00000100\ \ 00010000\ \ 00000000\ \ 10000010\ \ 00001000\ \ 00100000\ \ 0}1\\
s_4(t-1):&\ 0\underbrace{1000010\ \ 00010000\ \ 00000001\ \ 00001000\ \ 00000000\ \ 10000100\ \ 00100000\ \ 0}1\\
s_6(t-1):&\ 0\underbrace{1000000\ \ 00000000\ \ 00000000\ \ 00000000\ \ 00000000\ \ 00000000\ \ 00111111\ \ 1}1 \\ & \\
c(t):&\ \!*\!*011011\ \, 00010\!*\!01\ \ 000\!*\!000\!*\; 00101010\ \ 10101010\ \ 10001\!*\!10\ \  0\!*\!*11111\ *\!* \\
& \\
s_6(t):&\ \underbrace{10000000\ \ 00000000\ \ 00000000\ \ 00000000\ \ 00000000\ \ 00000000\ \ 01111111}
\end{align*}
\caption{CRT sequences in Example~\ref{ex:false_detection}. Sequence periods are indicated by underbraces.}
\label{fig:synch}
\end{figure*}

By comparing with $s_6(t)$, the receiver declares that the channel-activity signal is matched to $s_6(t)$ at time $0$. The receiver cannot distinguish whether user 6 starts transmitting at time 0 or~1, and
erroneously detects that user 6 becomes active at time~$0$, while the actual delay offsets is~1. Thus Algorithm~\ref{alg:active} fails in this case.
\label{ex:false_detection}
\end{example}

The preceding example illustrates that Algorithm~\ref{alg:active} does not work if  there are too many active users. Theorem~\ref{thm:synch} asserts that such error in synchronization does not occur if $q\geq 2p^2$ and the number of simultaneously active users does not exceed $(p+1)/2$.

If $q \equiv \pm 1 \bmod p$, the Hamming cross-correlation of the resulting CRT sequences is three-valued (see Theorem~\ref{thm:cross-correlation2}). For these preferred choices of $q$, we can improve Theorem~\ref{thm:synch}
by relaxing the requirement $q > 2p^2$ to $q > p^2$.

\begin{THEO}
 Let $p$ be a prime number, and $q > p^2$ be an integer such that $q \equiv \pm 1 \bmod p$.
Construct $p-1$ CRT sequences by Construction~\ref{construction:async} and assign them to users 1 to $p-1$. Then Algorithm~\ref{alg:active} can identify the active users and determine their starting time, provided that the number of simultaneously active users is no more than $(p+1)/2$.
\end{THEO}

The proof of Theorem~\ref{thm:synch}' is analogous to the proof of Theorem~\ref{thm:synch} in Appendix~\ref{app:proof_of_synch}.

\medskip

{\bf Erasure Correction and Throughput}

We have shown in Theorem~\ref{thm:synch} that the receiver is able to figure out the delay offsets of the active users. From a specific user's point of view, the channel reduces to an erasure channel.

In the remainder of this section, we pick $q$ to be an integer of the form $kp+1$, for some integer $k\geq p$, so that $q$ is larger than $p^2$ and Theorem~\ref{thm:synch}' can be applied.

For each user, the number of successfully received packets in a length is lower bounded by
\begin{equation}
(kp+1) - \Big( \frac{p+1}{2} - 1\Big)(k+1). \label{eq:erasure}
\end{equation}
The first term $kp+1$ is the total number of packets sent by a user in a period. The factor $k+1$ is the maximum Hamming cross-correlation from Theorem~\ref{thm:cross-correlation2}. After some simplifications, \eqref{eq:erasure} can be written as
\begin{equation}
D(p,k) := \frac{1}{2}[kp+k - p + 3].
\label{eq:rate}
\end{equation}

An erasure-correcting code can be applied across the packets in a period. We pick $Q$ to be a power of prime such that $Q \geq kp+1$. We encode $D(p,k)$ information packets using a shortened Reed-Solomon (RS) code of length $kp+1$ over the finite field of size~$Q$. The $kp+1$ encoded packets are sent out according to the assigned protocol sequence. Because of the maximal-distance separable property of RS codes, we can recover the information packets. We have thus proved the following

\begin{theorem}
Suppose that there are $(p+1)/2$ active users out of $p-1$ potential users in the collision channel without feedback, where $p$ is an odd prime. Each active user can send $D(p,k)$ information packets in a period  of $p(kp+1)$ time slots, where $D(p,k)$ is given in~\eqref{eq:rate} and $k$ is an integer larger than or equal to~$p$.
In particular, if we take $k = p$, the resultant CRT sequences have period $p^3+p$, and when $(p+1)/2$ users are active, the system throughput is lower bounded by
\[
\frac{p+1}{2} \cdot \frac{0.5 (p^2+3)}{p^3+p} \geq 0.25.
\]
\label{thm:throughput2}
\end{theorem}

\begin{example}
We pick $p=19$, $k=19$ and $q=p^2+1 = 362$ in Theorem~\ref{thm:throughput2}. Generate 18 CRT sequences of length $pq = 6878$ by Construction~\ref{construction:b}. Pick $Q = 512 = 2^9$, which is a prime power larger than the Hamming weight of the CRT sequences under consideration. Using a shortened RS code of length 362 and dimension $D(19,19) = 182$ over the finite field of size $512$, we encode 182 information packets in each period for each active user. When 10 users are active, the total number of information packets sent through the system is $182\times 10 = 1820$, achieving  system throughput no less than $1820/6878= 0.265$.
\end{example}

\section{Comparison with Other Protocol Sequences} \label{sec:uniform}

In order to compare the variation of Hamming cross-correlation due to delay offsets, we introduce in this section a measure of deviation called $\epsilon$-uniformity. 

Given any two binary and periodic sequences $a(t)$ and $b(t)$, let $\Ex_\tau[H_{ab}(\tau)]$ be the expectation of Hamming cross-correlation, with delay offset $\tau$ chosen uniformly at random over a period.

\begin{definition}
We say that the Hamming cross-correlation $H_{ab}$ is
{\em $\epsilon$-uniform} if
\begin{equation}
  \frac{ |H_{ab}(\upsilon) - \Ex_\tau[H_{ab}(\tau)]|}{\Ex_\tau[H_{ab}(\tau)]}  \leq \epsilon, \label{def:uniform}
\end{equation}
for all $\tau = 0,1,\ldots, L-1$. A sequence set is called {\em $\epsilon$-uniform} if $\epsilon$ is the smallest number such that each pair of distinct sequences is $\epsilon$-uniform. We say that a sequence set is {\em pairwise shift-invariant} if it is 0-uniform.
\end{definition}

In other words, $H_{ab}(\tau)$ is $\epsilon$-uniform if for all delay offsets $\tau$, the percentage difference between $H_{ab}(\tau)$ and the mean is between $-\epsilon$ and $\epsilon$. The notion of $\epsilon$-uniformity is the same as the normalized $\ell_\infty$ distance between Hamming cross-correlation and the expected value. 

The sum of the Hamming cross-correlation over all relative delay offsets in a period equals $w_a w_b$, where $w_a$ and $w_b$ denote the Hamming weight of $a(t)$ and $b(t)$ respectively (cf. Lemma~\ref{lemma:Pursley} in Appendix~\ref{app:proof_of_cross2}). If we take the average over all delay offsets~$\tau$, then
\[
 \Ex_\tau[H_{ab}(\tau)] = \frac{w_a w_b}{L}.
\]
Hence, the definition of $\epsilon$-uniformity in~\eqref{def:uniform} can be written as
\[
 \frac{| H_{ab}(\upsilon) - w_a w_b /L|}{w_a w_b / L }  \leq \epsilon.
\]

A lower bound on the worst-case throughput similar to~\eqref{eq:throughput} can be expressed in terms of $\epsilon$-uniformity. Suppose that there are $K$ active users and each of them is assigned a sequence from an  sequence set of length $L$ and Hamming weight~$w$. Suppose that the sequence set is $\epsilon$-uniform.
By the union bound, a user can successfully send at least
\[
 w - (K-1)(1+\epsilon)(w^2/L)
\]
packets in each period. Individual throughput is lower bounded by
\[
 f - (K-1)(1+\epsilon)f^2,
\]
where $f$ is the duty factor $w/L$. It can be easily seen that for fixed duty factor $f$ and number of users $K$, a smaller $\epsilon$ yields a larger lower bound on individual throughput.

\begin{example}
(Constant-weight cyclically permutable codes~\cite{Nguyen92}) In~\cite[p.948]{Nguyen92} Example~5, a protocol sequence set of period 156 and Hamming weight~12 is presented. The number of sequences is~169. It is shown that the Hamming cross-correlation is between 0 and 3, and the mean Hamming cross-correlation equals $(12^2)/156 = 12/13$. The maximal deviation from the mean is $3-12/13 = 27/13$. This sequence set is thus $(27/12)$-uniform.
\end{example}

\begin{example}
(Prime sequences~\cite{ShaarDavies83}) Given a prime $p$, we can construct a sequence set with period $p^2$, Hamming weight~$p$ and Hamming cross-correlation no more than 2. The mean Hamming cross-correlation is 1. For each pair of distinct prime sequences, the maximal $|H_{ab}(\upsilon) - \mathbb{E}_{\tau}[H_{ab}(\tau)]|$ over all possible delay offset~$\upsilon$ is equal to~1. The prime sequences are therefore 1-uniform.
\end{example}

\begin{example}
(Extended prime sequences~\cite{Yang95}) By padding $p-1$ zeros after every ``1'' in a prime sequence, we obtain a sequence set with period $p(2p-1)$, Hamming weight $p$ and Hamming cross-correlation either 0 or~1. The mean Hamming cross-correlation is $p/(2p-1)$, which is roughly equal to one half. We can check that the extended prime sequences are $1$-uniform.
\end{example}

\begin{example}
(Wobbling sequences~\cite{Wong07}) Based on prime sequences, a class of $(1/p)$-uniform sequence sets is constructed in~\cite{Wong07}. The number of sequences is $p$ and the sequence period is $p^4$.
\end{example}

\begin{example}
(Shift-invariant and pairwise shift-invariant sequences~\cite{Massey85,Rocha00,SCSW09,ZSW})
All pairwise shift-invariant is by definition 0-uniform. It is shown in~\cite{SCSW09} that the period grows exponentially in the number of sequences. For pairwise shift-invariant sequences, the period is also shown to grow exponentially with the number of users~\cite{ZSW}.
\end{example}

The examples above are presented in an order such that the $\epsilon$-uniformity is decreasing. We can see that the sequence period increases as we go down the list.

\begin{theorem}
The sequences obtained by Construction~\ref{construction:b}  with  period $O(kp^2)$ are $(1/k)$-uniform.
\end{theorem}

\begin{proof}
The Hamming cross-correlation in Construction~\ref{construction:b} is obtained by Theorem~\ref{thm:cross-correlation2}. The mean Hamming cross-correlation is
\[
 \frac{w^2}{L} = \frac{(kp-1)^2}{p(kp-1)} = \frac{kp-1}{p}.
\]
The maximal difference between Hamming cross-correlation and the mean is
\begin{align*}
\left| k \pm 1 - \frac{kp-1}{p} \right| &=
 \frac{1}{p}|\pm p -1|\\
& \leq \frac{1}{p}(p+1) \\
& =  \frac{kp-1}{p} \Big( \frac{p +1}{kp-1}  \Big)\\
& =  \frac{kp-1}{p} \,O\big(\frac{1}{k} \big).
\end{align*}
We thus have an $O(1/k)$-uniform sequence set of period $O(kp^2)$.
\end{proof}

\begin{table}
\begin{center}
\begin{tabular}{|c|c|c|} \hline
 & Period & $\epsilon$ \\ \hline \hline
Prime sequences & $p^2$ & 1 \\ \hline
Extended prime sequences & $p(2p-1)$ & 1 \\ \hline
CRT sequences & $O(kp^2)$ & $O(1/k)$ \\ \hline
Wobbling sequences & $p^4$ & $1/p$ \\ \hline
Shift-invariant sequences  & exponential in $p$ & 0\\ \hline
\end{tabular}
\end{center}
\caption{Tradeoff Between Period Length and $\epsilon$-uniformity  for Various Sequence Sets for $p$ Users}
\label{table:compare}
\end{table}

The trade-off between period length and $\epsilon$-uniformity is summarized in Table~\ref{table:compare}.

To compare with the wobbling sequences, we can take $k$ roughly equal to $p$. This yields CRT sequences with are $(1/p)$-uniform and period $O(p^3)$. The $\epsilon$-uniformity is roughly the same as wobbling sequences but the period  is shorter than the period of the wobbling sequences.

\section{Other Applications} \label{sec:discussions}

\subsection{Extension to Multiple Data Rates}

To support service with multiple data rates, we need protocol sequences with different duty factors; a sequence with larger duty factor is assigned to a user with higher data rate requirement. The CRT construction can be extended to cope with multiple data rates. For the sequence $s_{g,p,t}(t)$ generated by $g$, we replace the characteristic set $\set{I}_{g,p,q}$ by
\[
  \set{I}_{g,p,q} \cup (\set{I}_{g,p,q} +(1,0)) \cup \ldots \cup   (\set{I}_{g,p,q} + (k,0))
\]
for some positive integer $k$. We note that the above is a union of disjoint sets. The resulting sequence has duty factor $k/p$. The measure of uniformity however does not change; if the original CRT sequence set is $\epsilon$-uniform, the extended sequence set is still $\epsilon$-uniform.




\subsection{Application to Multi-Channel Network}
In this network, users can send data to each other, and we have a fully connected system topology.
The total bandwidth is divided into $p$ subchannels and each subchannel is assigned to a user. Let the subchannel assigned to user $i$ be denoted by subchannel $i$, for $i=1,2,\ldots, p$. Each user is also assigned a CRT sequence. The half-duplex model is assumed, so that each user cannot transmit and receive at the same time. Users always receive in their assigned subchannel. In a period of $L$ slots, user~$i$ can pick one user, say user $j$, and send packets to user~$j$ via subcarrier~$j$ using user~$i$'s CRT sequence. In one sequence period, user~$i$ either: (i) receives packets in subcarrier~$i$ for the whole period, or (ii) sends packets to user~$j$ in $L/p$ packets using subcarrier~$j$ and receives packets in the remaining $L-L/p$ packets in subcarrier~$i$.

The worst-case scenario occurs when user $i$ is sending packets to some other user, and all the remaining users want to send packets to user~$i$ in the same period. Using CRT sequences, we can show as in the multiple-access case that the worst-case throughput between each pair of users is lower bounded by a positive constant.

\section{Conclusion} \label{sec:conclusion}
A class of protocol sequences, called CRT sequences, whose Hamming cross-correlation is highly concentrated around the mean value is given.  When CRT sequences are applied to the collision channel without feedback, we obtain a tradeoff between worst-case throughput and the sequence period. The generation of CRT sequences involves only simple modular arithmetics, and provides a low-complexity solution to multiple-accessing in wireless sensor network.

\appendices

\section{Proof of Theorem~\ref{thm:cross-correlation1}}
\label{app:proof_of_cross1}

In this appendix,  $g$ denotes an element in $\mathbb{Z}_p \setminus \{0,1\}$, and
\begin{equation}
a_g(\tau_1,\tau_2) := (g-1)^{-1}(\tau_1 - \bar{\tau}_2) \bmod p.
\end{equation}
As defined in~\eqref{eq:bg}, $b_g$ denotes $(g-1)^{-1} \bar{q} \bmod p$. We will also use the indicator function $\mathbf{I}$ defined as
\[
\mathbf{I}(P) := \begin{cases}
1 & \text{if $P$ is true}, \\
0 & \text{if $P$ is false}.
\end{cases}
\]
Recall that $\bar{x}$ stands for the remainder of $x$ after division by~$p$.

\begin{lemma}  The Hamming cross-correlation $H_{g1}(\tau_1, \tau_2)$ as defined in~\eqref{def:H2D} satisfies the following properties:

\begin{enumerate}
\item $H_{g1}(\tau_1, \tau_2)$ equals the number of solutions to
\begin{equation}
\bar{x} \equiv a_g(\tau_1,\tau_2) \oplus_p b_g \mathbf{I}(0\leq x < \tau_2) \bmod p,
\label{eq:Ham4}
\end{equation}
for $x= 0,1,\ldots, q-1$.
\label{lemma:H2}

\item Let $(\tau_1, \tau_2)$ and $(\tau_1', \tau_2')$ denote two (2-dimensional) relative delay offsets. If $\tau_1 = \tau_1'$ and $\tau_2 = \tau_2' + kp$ for some integer $k$, then
  $$H_{g1}(\tau_1, \tau_2) = H_{g1}(\tau_1',\tau_2').$$
\label{lemma:H3}
\end{enumerate}
    \label{lemma:H}
\end{lemma}

\begin{proof}
After setting the value of $h$ in \eqref{eq:Ham} to 1, we obtain
\begin{equation}
 \bar{x} g \equiv  (\overline{x \ominus_q \tau_2}) \oplus_p \tau_1 \bmod p.
 \label{eq:Ham_app}
\end{equation}
We want to show that the number of solutions to~\eqref{eq:Ham_app}, for $x=0,1,\ldots, q-1$, is the same as the number of solutions to~\eqref{eq:Ham4}.

Consider $x$ in two disjoint ranges: (i) $0 \leq x <\tau_2$, and (ii) $\tau_2 \leq x < q$. In the first case, $x \ominus_q \tau_2$ is congruent to $x+q-\tau_2 \bmod q$. So, for $0\leq x <  \tau_2$, \eqref{eq:Ham_app} is equivalent to
\begin{equation}
 \bar{x}g \equiv \bar{x}\oplus_p \bar{q}\ominus_p \bar{\tau}_2 \oplus_p \tau_1 \bmod p
\label{eq:Ham2}
\end{equation}
where  $\bar{q}$ and $\bar{\tau}_2$ are residues of  $q$ and $\tau_2$ in $\mathbb{Z}_p$, respectively.

In the second case, for $x=\tau_2, \tau_2+1, \ldots, q-1$, \eqref{eq:Ham_app} is equivalent to
\begin{equation}
\bar{x}g \equiv \bar{x}\ominus_p\bar{\tau}_2\oplus_p\tau_1 \bmod p.
\label{eq:Ham3}
\end{equation}

We combine \eqref{eq:Ham2} and \eqref{eq:Ham3} in one line as
\[
\bar{x}(g-1) \equiv  -\bar{\tau}_2\oplus_p \tau_1 \oplus_p \bar{q} \mathbf{I}(0\leq x < \tau_2) \bmod p.
\]
Since $g$ is not equal to 1 by assumption, we can divide by $(g-1)$ and obtain~\eqref{eq:Ham4}. This proves the first part of the lemma.

\medskip

The second part of the lemma is vacuous if $q<p$. So we assume $q>p$. (The case $q=p$ is excluded because it is assumed that $q$ is relatively prime with $p$.)
It is sufficient to prove the statement for $\tau_1 = \tau_1'$ and $\tau_2' = \tau_2+p$, namely, the number of solutions to~\eqref{eq:Ham4} and the number of solutions to
\begin{equation}
\bar{x} \equiv a_g(\tau_1, \tau_2+ p) + b_g \mathbf{I}(0\leq x <\tau_2+p) \bmod p \label{eq:LemmaB}
\end{equation}
for $x=0,1,\ldots, q-1$, are the same. We note that $a_g(\tau_1, \tau_2)$ is equal to $a_g(\tau_1, \tau_2+p)$.  However, the arguments inside the indicator function are different. We divide the range of $x$ into three disjoint parts:
\begin{align*}
\set{X}_1 &:= \{0,1,\ldots, \tau_2-1\}, \\
\set{X}_2 &:= \{\tau_2,\tau_2+1,\ldots, \tau_2+p-1\}, \\
\set{X}_3 &:= \{\tau_2+p,\tau_2+p+1,\ldots, q-1\}.
\end{align*}
For $x \in \set{X}_1$, $$\mathbf{I}(0\leq x < \tau_2) = \mathbf{I}(0\leq x < \tau_2+p)= 1.$$
Therefore
\eqref{eq:Ham4} and~\eqref{eq:LemmaB} have the same number of solutions for $x$ in $\set{X}_1$. For $x\in \set{X}_2$, both \eqref{eq:Ham4} and \eqref{eq:LemmaB}
have exactly one solution by Lemma~\ref{lemma:simple}.
For $x \in \set{X}_3$, we have
$$\mathbf{I}(0\leq x < \tau_2) = \mathbf{I}(0\leq x < \tau_2+p)= 0,$$
and hence \eqref{eq:Ham4} and \eqref{eq:LemmaB} have the same number of solutions for $x \in \mathcal{X}_3$. In conclusion, the number of solutions to \eqref{eq:Ham4} and \eqref{eq:LemmaB} for $x \in\mathcal{X}_1\cup\mathcal{X}_2\cup\mathcal{X}_3$ are the same. This finishes the proof of the second part of the lemma.
\end{proof}

\begin{proof}[Proof of Theorem~\ref{thm:cross-correlation1}]
By the second part of the previous lemma, we only need to consider $\tau_2 = 0,1,\ldots, p-1$.
We first consider the case when $b_g$ is between 1 and $p-\bar{q}-1$. We further consider two subcases.

{(i)  $0\leq \tau_2 < \bar{q}$:}

Suppose  that~\eqref{eq:Ham4} has no solution for $0\leq x< \tau_2$. As the indicator function in~\eqref{eq:Ham4} is zero for $x = \tau_2$, $\tau_2+1, \ldots, q-1$, \eqref{eq:Ham4} is reduced to
\[
 \bar{x} \equiv a_g(\tau_1, \tau_2) \bmod p.
\]
The number of integers in $\{ \tau_2$, $\tau_2+1, \ldots, q-1\}$, say $d$, satisfies $\lfloor d/p \rfloor = m$.
By Lemma~\ref{lemma:simple}, we have either $m$ or $m+1$ solutions to~\eqref{eq:Ham4}  for $x \geq \tau_2$.

Secondly, suppose that~\eqref{eq:Ham4} has exactly one solution for $0\leq x< \tau_2$. The indicator function in~\eqref{eq:Ham4} is equal to~1 for $0\leq x < \tau_2$. Hence,
\begin{equation}0\leq a_g(\tau_1,\tau_2)+b_g < \tau_2. \label{eq:contradiction1}
 \end{equation}
We claim that \eqref{eq:Ham4} has no solution for $x=\tau_2, \tau_2+1, \ldots, \bar{q}-1$. Otherwise, we have
\[
\tau_2 \leq a_g(\tau_1,\tau_2) < \bar{q},
\]
which, after combining with the assumption that $1 \leq b_g \leq p - \bar{q}-1$, yields
\[
\tau_2 < a_g(\tau_1,\tau_2) + b_g < p-1.
\]
This contradicts with~\eqref{eq:contradiction1} and proves the claim. For $$\bar{q}\leq x < q,$$ there are exactly $m$ solutions by Lemma~\ref{lemma:simple}. The total number of solutions to~\eqref{eq:Ham4} for $x=0,1,\ldots, q-1$, is thus $m+1$. Hence $H_{g1}(\tau_1, \tau_2) = m+1$.

{(ii)  $\bar{q} \leq \tau_2 < p$:}

By Lemma~\ref{lemma:simple}, \eqref{eq:Ham4} has either 0 or 1 solution for $0\leq x \leq \tau_2$, and either $m-1$ or $m$ solutions for $\tau_2 \leq x < q$. Hence, $H_{g1}(\tau_1, \tau_2)$ is within the range of $\{m-1, m, m+1\}$.

\medskip

For $b_g = p-\bar{q}+1, \ldots, p-1$, we again consider two subcases.

{(i) $0 \leq \tau_2 < \bar{q}$:}

By Lemma~\ref{lemma:simple}, \eqref{eq:Ham4} has either 0 or 1 solution for $0\leq x < \tau_2$, and either $m$ or $m+1$ solutions for $\tau_2 \leq x < q$. Therefore, $H_{g1}(\tau_1, \tau_2) \in \{m, m+1, m+2\}$.

{(ii)  $\bar{q} \leq \tau_2 < p$:}

Suppose that \eqref{eq:Ham4} has no solution for $0 \leq x < \tau_2$, i.e.,
\begin{equation}
\tau_2 \leq a_g(\tau_1, \tau_2) + b_g < p. \label{eq:contradiction2}
\end{equation}
We claim that \eqref{eq:Ham4} must have one solution for $x$ in the following range
\begin{equation}
\tau_2 \leq x < p+\bar{q}. \label{eq:range}
\end{equation}
From the assumption of $\bar{q} \leq \tau_2 < p$, we deduce that
$$\bar{q}<p+\bar{q}-\tau_2 \leq p,$$ so that the range in~\eqref{eq:range} is non-empty and consists of no more than $p$ integers.
If the claim were false, we would have no solution to~\eqref{eq:Ham4} for $\tau_2 \leq x < p+\bar{q}$, implying that
\begin{equation}
\bar{q} \leq a_g(\tau_1, \tau_2) < \tau_2. \label{eq:appA}
\end{equation}
Here, we have used the fact that the indicator function in~\eqref{eq:Ham4} is equal to zero for $x$ in the range in~\eqref{eq:range}.
By adding~\eqref{eq:appA} to $$p-\bar{q}+1\leq b_g \leq p-1$$ and reducing mod $p$, we obtain
\[
1 \leq  a_g(\tau_1, \tau_2) + b_g < \tau_2,
\]
which is a contradiction to~\eqref{eq:contradiction2}. Thus, the claim is proved. For $x = p+\bar{q}, p+\bar{q}+1, \ldots, q-1$, there are exactly $m-1$ solutions to~\eqref{eq:Ham4} by Lemma~\ref{lemma:simple}. Totally there are $m$ solutions, and thus $H_{g1}(\tau_1,\tau_2) = m$.

Finally suppose that \eqref{eq:Ham4} has exactly one solution for $0 \leq x < \tau_2$. As the number of solutions to~\eqref{eq:Ham4} for $x=\tau_2, \tau_2+1, \ldots, q-1$ is either $m-1$ or $m$ by Lemma~\ref{lemma:simple}, the total number of solutions to~\eqref{eq:Ham4} is either $m$ or $m+1$.

In any case, we see that $H_{g1}(\tau_1, \tau_2)$ is either $m$, $m+1$ or $m+2$.
\end{proof}

\section{Proof of Theorem~\ref{thm:frequency}}
\label{app:proof_of_cross2}

We use the following property of Hamming cross-correlation which holds for any binary sequence set in general~\cite{Pursley80}.
\begin{lemma} Let $a(t)$ and $b(t)$ be binary sequences with period $L$ and Hamming weight~$q$. Then
\[\sum_{\tau=0}^{L-1} H_{ab}(\tau) = q^2.\] \label{lemma:Pursley}
\end{lemma}

We include the short proof for completeness.

\begin{proof}
\begin{align*}
 \sum_{\tau=0}^{L-1} H_{ab}(\tau)
&=  \sum_{\tau=0}^{L-1}  \sum_{t=0}^{L-1} a(t) b(t-\tau) \\
&= \sum_{t=0}^{L-1} a(t) \sum_{\tau=0}^{L-1} b(t-\tau) \\
&= \sum_{t=0}^{L-1}  a(t) \sum_{\tau=0}^{L-1} b(\tau) = q^2.
\end{align*}
The last equality follows from the assumption that the Hamming weights of $a(t)$ and $b(t)$ are both~$q$.
\end{proof}

1) By Theorem~\ref{thm:cross-correlation0}, $N_0(j)$ is nonzero only when $j=m$ or $j=m+1$, whence,
\begin{equation}
N_0(m)+N_0(m+1) = pq. \label{eq:app1}
\end{equation}
On the other hand, we have
\begin{equation}
m N_0(m) + (m+1) N_0(m+1) = q^2 \label{eq:app2}
\end{equation}
from Lemma~\ref{lemma:Pursley}. We can eliminate $N_0(m)$ from~\eqref{eq:app1} and~\eqref{eq:app2} and obtain
\[m pq + N_0(m+1) = q^2,\]
which implies $N_0(m+1) = q(q-mp) = q \bar{q}$. Substituting this into~\eqref{eq:app1}, we get $N_0(m) = pq - N_0(m+1) = (p-\bar{q})q$. This proves the first part of Theorem~\ref{thm:frequency}.

\medskip

2) Suppose that $b_g=1,2,\ldots, p-\bar{q}-1$. We can set up a system of two linear equations in three variables $N_g(m-1)$, $N_g(m)$ and $N_g(m+1)$:
\begin{align*}
\sum_{k=m-1}^{m+1} N_g(k) &= pq \\
\sum_{k=m-1}^{m+1}  kN_g(k) &= q^2.
\end{align*}
The second equality is due to Lemma~\ref{lemma:Pursley}.
Solving for $N_g(m)$ and $N_g(m+1)$ in terms of $N_g(m-1)$, we obtain~\eqref{eq:dist1} to~\eqref{eq:dist3}. It remains to  evaluate $N_g(m-1)$. The proof is completed by showing the following two claims:

(i) For each $k=0,1,\ldots, m-1$, there are exactly
\begin{equation}
b_g(p - b_g- \bar{q}) \label{eq:orderpairs1}
\end{equation}
order pairs $(\tau_1, \tau_2)$, with $0 \leq \tau_1 < p$ and $\tau_2 = kp,kp+1, \ldots, (k+1)p-1$, such that $H_{g1}(\tau_1, \tau_2) = m-1$.

(ii) For $\tau_2 = mp, mp+1, \ldots, q-1$, $H_{g1}(\tau_1, \tau_2)$ does not equal $m-1$ for all $0 \leq \tau_1 < p$.

Multiplying~\eqref{eq:orderpairs1} by $m$, we obtain
\[N_g(m-1) = m b_g(p - b_g- \bar{q}).\]
In the following, we complete the proof of part 2 in the theorem by showing (i) and~(ii).

By  Lemma~\ref{lemma:H} part~\ref{lemma:H3}, we notice that $H_{g1}(\tau_1, \tau_2)$ depends on $\tau_2$ only through the residue of $\tau_2 \bmod p$. Hence it is sufficient prove claim (i) for $k=0$.

Consider $\tau_2$ from 0 to $p-1$. We want to count the number of times that $H_{g1}(\tau_1, \tau_2) = m-1$, for $0\leq \tau_1, \tau_2 < p$. We have shown in Lemma~\ref{lemma:H} that $H_{g1}(\tau_1, \tau_2)$ can be computed by counting the number of solutions to~\eqref{eq:Ham4} for $x$ between 0 and $q-1$. Partition the range of $x$ into two disjoint subsets
\begin{align*}
 \set{X}_1 &:= \{0,1,\ldots, p+\bar{q}-1\}, \text{ and} \\
 \set{X}_2 &:= \{ p+\bar{q}, p+\bar{q}+1, \ldots, q-1\},
\end{align*}
and consider the number of solutions to~\eqref{eq:Ham4} for $x$ in $\mathcal{X}_1$ and~$\mathcal{X}_2$ separately.
Because $0\leq \tau_2 < p$, the indicator function $\mathbf{I}(0\leq x < \tau_2)$ in~\eqref{eq:Ham4} is identically equal to 0 for $x \in \set{X}_2$. The number of solutions to~\eqref{eq:Ham4} for $x \in \set{X}_2$ is exactly $m-1$ by Lemma~\ref{lemma:simple}. The problem reduces to counting the number of pairs $(\tau_1, \tau_2) \in \mathbb{Z}_p^2$ such that \eqref{eq:Ham4} has no solution for $x \in \set{X}_1$.
Observe that $a_g$ depends on $\tau_1$ and $\tau_2$ through their difference $\tau_1-\tau_2 \bmod p$. We make a change of variables $(\tau_1, \tau_2) \rightarrow (u, \tau_2)$ by defining $u$ as
\begin{equation}
u := (g-1)^{-1}(\tau_1-\tau_2) \bmod p. \label{def:u}
 \end{equation}
Our objective now is to count the number of ordered pairs $(u,\tau_2) \in \mathbb{Z}_p^2$ such that the equation
\begin{equation}
 x \equiv u + b_g \mathbf{I}(0\leq x < \tau_2) \bmod p
 \label{eq:modified}
\end{equation}
has no solution for $x\in \set{X}_1$. We note that $b_g$ does not depend on $\tau_2$ and~$u$.

If
\begin{equation}0 \leq u < \bar{q}, \label{eq:range_of_u_1}
\end{equation}
then \eqref{eq:modified} has at least one solution over $x\in \set{X}_1$, namely $x = ( u \bmod p) + p$. Indeed, as $\tau_2 < p \leq x$, the indicator function $\mathbf{I}(0\leq x < \tau_2)$ is evaluated to 0, and thus if we put $x = ( u \bmod p) + p$ in~\eqref{eq:modified}, we have
$( u \bmod p) + p \equiv u \bmod p$, which obviously holds.

On the other hand, if
\begin{equation} p-b_g \leq u < p, \label{eq:range_of_u_2}
\end{equation}
then~\eqref{eq:modified} also has at least one solution over $x \in \set{X}_1$ no matter what $\tau_2$ is. Indeed, for $0 \leq \tau_2 \leq u$, we can set $x=u$. Then $\mathbf{I}(0\leq x < \tau_2)$ equals 0, and we see that $x=u$ is a solution to~\eqref{eq:modified}. For $u < \tau_2 < p$, we can set $x = u+b_g - p$. Then we have $x < u \leq \tau_2$ and whence $\mathbf{I}(0\leq x < \tau_2)=1$. When $u < \tau_2 <p$ and $x= u+b_g-p$, \eqref{eq:modified} becomes
\[
 u + b_g - p \equiv u + b_g \bmod p.
\]

From \eqref{eq:range_of_u_1} and~\eqref{eq:range_of_u_2}, $H_{g1}(\tau_1, \tau_2)$ is equal to $m-1$ only when  $u = \bar{q}, \bar{q}+1, \ldots, p-b_g-1$.
For each such $u$, we now count the number of $\tau_2\in\mathbb{Z}_p$ for which~\eqref{eq:modified} has no solution over $x\in \set{X}_1$. If $x$ is a solution of~\eqref{eq:modified}, then $x$ can take only two values, namely $u$ or $u+b_g$. When $x=u$ is a solution, $x$ must satisfy $x \geq \tau_2$;  when $x=u+b_g$ is a solution, $x$ must satisfy $0\leq x < \tau_2$. So, \eqref{eq:modified} has no solution over $x\in\set{X}_1$ if and only if for all $x\in \set{X}_1$, we have $u < \tau_2$ and $u+b_g \geq \tau_2$. Putting these two inequalities together, we obtain $u < \tau_2 \leq u+b_g$. Consequently, for each $u = \bar{q}, \bar{q}+1, \ldots, p-b_g-1$, there are exactly $b_g$ values of $\tau_2$ such that $H_{g1}(\tau_1, \tau_2) = m-1$. This proves claim (i).

For claim (ii), we count the solutions to~\eqref{eq:modified} for $0\leq x \tau_2$. Since $\mathbf{I}(0\leq x < \tau_2)$ is identically equal to~1, by Lemma~\ref{lemma:simple}, the number of solutions to~\eqref{eq:modified} over $0\leq x < \tau_2$ is at least~$m$. So $H_{g1}(\tau_1,\tau_2)$ cannot be $m-1$.

This ends the proof of the second part of Theorem~\ref{thm:frequency}.

\medskip

3) We  set up the following system of linear equations in variables $N_g(m)$, $N_g(m+1)$ and $N_g(m+2)$:
\begin{align*}
\sum_{k=m}^{m+2} N_g(k)&= pq \\
\sum_{k=m}^{m+2}  kN_g(k)&= q^2.
\end{align*}
After solving for $N_g(m)$ and $N_g(m+1)$
in terms of $N_g(m+2)$,  we obtain \eqref{eq:dist4} to \eqref{eq:dist6}. So, we just need to evaluate $N_g(m+2)$.

The evaluation of $N_g(m+2)$ relies on the following two claims:

(i) For each $k=0,1,2,\ldots,m$, there are exactly
\begin{equation}
(p-b_g)(\bar{q}+b_g-p) \label{eq:app_claim2}
\end{equation}
ordered pairs $(\tau_1, \tau_2)$, with $0\leq \tau_1 < p$ and $kp \leq \tau_2 < kp+\bar{q}$, such that $H_{g1}(\tau_1, \tau_2) = m+2$.

(ii) For $k=0,1,\ldots, m-1$, none of the ordered pair $(\tau_1,\tau_2)$ with $0\leq \tau_1 < p$ and $kp+\bar{q} \leq \tau_2 < (k+1)p$, satisfies $H_{g1}(\tau_1, \tau_2) = m+2$.

Since (i) and (ii) exhaust all possible $(\tau_1, \tau_2) \in G_{p,q}$, we multiply~\eqref{eq:app_claim2} by $(m+1)$ and obtain
\[N_g(m+2) = (m+1)(p-b_g)(\bar{q}+b_g-p).\]
We prove case (i) and (ii) in the rest of this appendix.

By the second part of Lemma~\ref{lemma:H3}, it is sufficient to prove case (i) for $k=0$. Consider $\tau_2$ from 0 to $\bar{q}-1$. Divide the range of $x$ into two disjoint parts:
\begin{align*}
 \set{X}_3 &:= \{0,1,\ldots, \bar{q}-1\}, \text{ and} \\
 \set{X}_4 &:= \{\bar{q}, \bar{q}+1, \ldots, q-1\}.
\end{align*}
As in the proof of the second part of this theorem, we make a change of variable $(\tau_1, \tau_2) \rightarrow (u, \tau_2)$ by defining $u$ as in~\eqref{def:u}. It is noted that $\set{X}_4$ consists of $mp$ consecutive integers, and the indicator function $\mathbf{I}(0 \leq x < \tau_2)$ in~\eqref{eq:modified} equals 0 for all $x \in\set{X}_4$. By Lemma~\ref{lemma:simple}, there are exactly $m$ solutions to~\eqref{eq:modified} for $x \in \set{X}_4$.  So, $H_{g1}(\tau_1,\tau_2)$ equals $m+2$ if and only if \eqref{eq:modified} has exactly two solutions over $x \in \set{X}_3$. It reduces the problem to counting the number of pairs $(u,\tau_2)$, with $0\leq u <p$ and $0 \leq \tau_2 < \bar{q}$, such that~\eqref{eq:modified} has exactly two solutions over $x \in \set{X}_3$.

The only two candidate solutions for \eqref{eq:modified} are
$x=u$ and $x=u+b_g$. We investigate under what condition both $u$ and $u+b_g$ are indeed solutions to~\eqref{eq:modified}. $x=u$ is a solution only if $\mathbf{I}(0 \leq x < \tau_2) = 0$. This implies that $u < \tau_2$. $x=u+b_g$ is a solution only if $\mathbf{I}(0 \leq x < \tau_2) = 1$. Hence $u + b_g < \tau_2$. Combining these two conditions, we obtain
\[
 u+b_g \leq  \tau_2 < u.
\]
This is possible only if $u+b_g \geq p$, and thus after reduction mod $p$, we have $(u + b_g \bmod p) < u$. The first necessary condition for both $u$ and $u+b_g$ are solutions to~\eqref{eq:modified} is
\begin{equation}
 p - b_g \leq u < p. \label{eq:range_of_u_3}
\end{equation}

Secondly, as it is required that $x \in \set{X}_3$, we must have $u \in \set{X}_3$, i.e.,
\begin{equation}0 \leq u < \bar{q} \label{eq:range_of_u_4}
\end{equation}

Putting \eqref{eq:range_of_u_3} and~\eqref{eq:range_of_u_4} together, we have the following necessary condition on~$u$.
\begin{equation}
 p- b_g \leq u <\bar{q}. \label{eq:range_of_u_5}
\end{equation}
We note that the range of $u$ in~\eqref{eq:range_of_u_5} is nonempty, because $p - b_g < \bar{q}$ by assumption.
$u$ can take on any of the $(b_g + \bar{q} - p)$ value in~\eqref{eq:range_of_u_5}, and for each such $u$, $\tau_2$ can assume values in $u-\{0,1,\ldots, p-b_g-1\}$. Hence, the total number of pairs $(\tau_1, \tau_2)$ such that $0\leq \tau_2 < \bar{q}$ and $H_{g1}(\tau_1,\tau_2) = m+2$ is $(p-b_g)(b_g - (p-\bar{q}))$.

For case (ii), we again use the fact that $H_{g1}(\tau_1, \tau_2)$ depends on $\tau_2$ only through the residue of $\tau_2$ mod $p$, and establish case (ii) only for $k=0$. Let $\tau_2$ be in the range $\bar{q} \leq \tau_2 < p$. Consider the solutions to~\eqref{eq:modified} separately in $0 \leq x < \tau_2$ and $\tau_2 \leq x < q$. For $0 \leq x < \tau_2$, $\mathbf{I}(0 \leq x < \tau_2)$ is identically equal to~1. By Lemma~\ref{lemma:simple}, there are at most 2 solutions to~\eqref{eq:modified} for $x = 0, 1, \ldots,\tau_2-1$. For $\tau_2 \leq x < q$, $\mathbf{I}(0 \leq x < \tau_2)$ is identically equal to~0, and by Lemma~\ref{lemma:simple}, there are at most $m-1$ solutions to~\eqref{eq:modified} for $\tau_2 \leq x < q$. There are totally at most $m+1$ solutions to~\eqref{eq:modified}.

This completes the proof of Theorem~\ref{thm:frequency}.

\begin{figure*}
\begin{center}
\begin{tabular}{|*{19}{c}|} \hline
{\bf 0}&     {\bf 3}&     {\bf 6}&     9&    12&    15&    18&    21&    24&    27&    30&    33&    36&    39&    42&    45&    48&    51&    54 \\
19&    22&    25&    28&    31&    34&    37&    40&    43&    46&    49&    52&   55&    {\bf 1}&     {\bf 4}&     {\bf 7}&    10&    13&    16 \\
38&    41&    44&    47&    50&    53&    56&     {\bf 2}&     {\bf 5}&     {\bf 8}&    11&    14&    17&    20&    23&    26&    29&    32&    35 \\ \hline
\end{tabular}
\end{center}
\caption{Mapping from $\mathbb{Z}_{57}$ to a $3 \times 19$ array $\mat{M}_{3,19}$. The numbers 0 to 8 are highlighted} \label{fig:57}
\end{figure*}

\section{Proof of Theorem~\ref{thm:synch}}
\label{app:proof_of_synch}

In this proof, $p$ is a prime number and $q$ is an integer relatively prime to $p$ and strictly larger than $2p^2$, $L=pq$ is the sequence period, and $\gamma$ is the multiplicative inverse of $p \bmod q$, i.e., $\gamma p \equiv  1 \bmod q$.  The unique integer between 0 and $q-1$ which is equal to $x$ mod $q$ is denoted by$(x \bmod q)$. The translate of a subset $\set{S}$ in $G_{p,q}$ by $(\tau_1, \tau_2)$ is defined as
\[
 \set{S} + (\tau_1,\tau_2) := \{ (x,y)+(\tau_1,\tau_2):\,
 (x,y) \in \set{S}\}.
\]

Consider user~$g$, where $i=1,2,\ldots, p-1$.
If user $g$ starts transmitting at time $t_0$, then the channel-activity signal is matched to $s_g(t)$ at time~$t_0$. The receiver will never fail to detect the presence of user~$g$, meaning that if user $g$ does start transmitting, the receiver can always detect this change of status from idle to active. The only sources of error are (a) detecting a user but in fact that user is not active, and (b) miscalculation of the start time. We refer to the error in (a) as {\em false alarm} and (b) as {\em synchronization error}.

We now show that false alarm cannot occur. Suppose on the contrary that the channel-activity signal is matched to $s_g(t)$ at time~$t_0$, but user~$g$ is idle from time $t_0$ to~$t_0+L-1$.  If this happened, the $q$  time indices in $\set{I}_g + t_0$ would be covered by the protocol sequences of other active users. However, the cross-correlation between user $g$ and each other active user is upper bounded by $\lfloor q/p \rfloor +2$, by Theorem~\ref{thm:cross-correlation1}. Because user~$g$ is assumed to be inactive in this period, the number of simultaneously active users does not exceed the maximum $(p+1)/2$, and hence the number of time slots in $\set{I}_{g}+t_0$ with a packet or collision observed is no larger than
\begin{align}
\big(\lfloor q/p \rfloor +2 \big) \Big( \frac{p+1}{2}\Big) &<
\Big(\frac{q}{p} +2 \Big) \Big( \frac{p+1}{2}\Big) \notag \\
&  = q\Big(\frac{1}{p} + \frac{2}{q} \Big) \frac{p+1}{2} \notag \\
& < q \Big(\frac{1}{p} + \frac{2}{2p^2}  \Big) \frac{p+1}{2}  \label{eq:false_alarm}\\
&=  \frac{q}{2} \Big(\frac{p+1}{p}\Big)^2 \notag
\end{align}
In~\eqref{eq:false_alarm}, we have used the assumption that $q > 2p^2$. We note that for all $p \geq 3$, the factor $(p+1)^2/p^2$ is strictly less than two. We obtain
\[
\big(\lfloor q/p \rfloor +2 \big) \Big( \frac{p+1}{2}\Big) < q.
\]
But $q$ is precisely the total number of ones in a period of $s_g(t)$. The $q$ time slots indices by $\set{I}_g+t+0$ cannot be covered by any other $(p+1)/2$ CRT protocol sequences. The channel-activity signal $c(t)$ cannot be matched to $s_g(t)$, and therefore false alarm cannot occur.

For synchronization error, assume that user $g$ is idle from time $t_0-L+1$ to $t_0-1$, and becomes active at time $t_0$. Our objective is to show that the channel-activity signal is not matched to $s_g(t)$ at time $t_0-\tau$, for any integer $\tau$ between 1 and $L-1$. The idea of showing that synchronization error cannot occur is the following.
If the channel-activity signal were matched incorrectly to $s_g(t)$ at $t_0-\tau$, then the receiver would observe $q$ ``1'' or ``*'' at time slots indexed by $t_0-\tau+\Phi_{p,q}'^{-1}(\set{I}_{g,p,q})$. Among these $q$ time slots, say $b$ of them come from a shifted version of $s_g(t)$, starting at time $t_0$. We then show that the remaining $q-b$ slots cannot be covered by the other active users. We divide the proof into several propositions below.

In the modified CRT correspondence $\Phi'$, an element $t \in \mathbb{Z}_{pq}$ is mapped to $(t \bmod p, \gamma t \bmod q)$.
In order to visualize the mapping, we introduce a matrix $\mat{M}_{p,q}$.
\begin{definition}
Given relatively prime integers $p$ and $q$, let $\mat{M}_{p,q}$ be a $p\times q$ matrix whose $(i,j)$-entry equals $t$ if $i\equiv t \bmod p$ and $j\equiv \gamma t \bmod q$,
for $t=0,1,\ldots, pq-1$. The rows and columns of $\mat{M}_{p,q}$ are indexed by $\{0,1,\ldots, p-1\}$ and $\{0,1,\ldots, q-1\}$, respectively.
\end{definition}

Each integer from 0 to $pq-1$ appears exactly once in $\mat{M}_{p,q}$. An example for $p=3$ and $q=19$ is shown in Fig.~\ref{fig:57}. We pay special attention to the integers from 0 to $p^2-1$, and want to get a handle on where they are located in $\mat{M}_{p,q}$. In Fig.~\ref{fig:57}, we can see that 0, 3 and 6 are on the upper left corner of $\mat{M}_{3,19}$. The numbers 1, 4 and 7 occupy three consecutive entries in row~1. The numbers 2, 5 and 8 occupy three consecutive entries in row~2.

\begin{prop}
(i) For any $i$ and $j$, the $(i,j+1)$-entry is equal to $p$ plus the $(i,j)$-entry mod $pq$.

(ii) Under the modified CRT correspondence $\Phi_{p,q}'$, the integers $kp$, for $k=0,1,2,\ldots, q-1$, are mapped to $(0, k)$. They appear in the first row of $\mat{M}_{p,q}$.

(iii) The numbers from 1 to $p^2-1$, except the multiples of $p$, are located between column $2p+1$ and column $q-p-1$ inclusively in $\mat{M}_{p,q}$.
\label{prop:M}
\end{prop}

\begin{proof}
(i) The $(0,1)$-entry in $\mat{M}_{p,q}$ is labeled by~$p$, because
\[
 p \mapsto (p \bmod p, \gamma p \bmod q) = (0,1),
\]
and  $\gamma p \equiv 1 \bmod q$ by the defining property of $\gamma$.

(ii) For $k=0,1,2,\ldots, q-1$,
\[
\Phi_{p,q}'(kp) = (kp, kp\gamma) = (0, k(1) ) = (0,k).
\]

(iii) We have the following claim:
\begin{equation}
 2p < (k\gamma \bmod q) \leq q-2p  \label{eq:claim_range}
\end{equation}
for $k=1,2,\ldots, p-1$.

We prove the claim by contradiction. Suppose that $k\gamma$, after reduction mod $q$, is between 1 and $2p$. Then, $k\gamma p \bmod q$ is equal to $p$, $2p$, $3p, \ldots,$ or $2p^2$. Since $q > 2p^2$, these $p$ numbers remain unchanged after reduction mod $q$. However,
$k\gamma p \equiv k(\gamma p) \equiv k \bmod q$, and this contradicts the assumption that $k$ is between 1 and $p-1$. Now suppose that $(k\gamma \bmod q)$ is equal to $q - 2p+1$, $q-2p+2, \ldots,$ or $q-1$. Then, the value of $k \gamma p$, after reduction mod $q$, is equal to
$$q-2p^2+p,\ q- 2p^2+2p, \ldots,\text{ or } q- p.$$
Since $k\gamma p \equiv k \bmod q$, this also contradicts  that $k$ is between 1 and $p-1$. This finishes the proof of the claim.

Let $\ell$ be an integer between 1 and $p^2-1$ which is not a multiple of $p$. We can write $\ell$ as $mp+k$ for some $m$ and $k$ between 1 and $p-1$. By part (i), the location of $\ell$ in $\mat{M}_{p,q}$ is $m$ steps to the right of the location of $k$ in $\mat{M}_{p,q}$. But $k$ cannot be located to the right of column $q-2p$. The right-most column in $\mat{M}_{p,q}$ which may contain $\ell$ is thus $q-p-1$.
This finishes the proof of part (ii).
\end{proof}

\begin{prop}
(i) Let $\set{S} = \{0, 1, 2,\ldots, p^2-1\}$, and $\set{S}'$ be the image of $\set{S}$ under $\Phi'_{p,q}$, i.e., $\set{S}' = \Phi_{p,q}'(\set{S})$. We have
\begin{equation}
 |\set{I}_g \cap (\set{I}_h+(\tau_1, \tau_2)) \cap \set{S}' | \leq 2 \label{eq:partial}
\end{equation}
for any given $(\tau_1,\tau_2)$ and $g\neq h$.

(ii) For $g=1,2,\ldots, p-1$, there are exactly $p$ ones in the first $p^2$ bits of $s_g(t)$, i.e., there are exactly $p$ ones among $s_g(0)$, $s_g(1), \ldots, s_g(p^2-1)$.
\label{prop:aperiodic}
\end{prop}

The quantity on the left hand side of~\eqref{eq:partial} can be interpreted as the {\em partial Hamming cross-correlation}, defined as
\[
  \sum_{t=0}^{p^2-1} s_g(t) s_h(t+\tau).
\]
We only consider the number of overlaps in the first $p^2$ time indices. The proposition asserts that the partial Hamming cross-correlation of two CRT sequences cannot exceed two.

\begin{proof}
(i) By part (i) of Prop.~\ref{prop:M}, for each $k$ between 0 and $p-1$, the following $p$ integers,
\[
 k, k+p, k+2p, \ldots, k+(p-1)p,
\]
occupy $p$ consecutive horizontal entries in the $k$th row of $\mat{M}_{p,q}$. Wrapping around the right boundary of $\mat{M}_{p,q}$ is precluded by Prop.~\ref{prop:M}.

A common element of $\set{I}_g$ and $\set{I}_h+(\tau_1,\tau_2)$ is in the form
\begin{equation}
 (g t_1, t_1) = (h t_2 + \tau_1, t_2 + \tau_2) \label{eq:common}
\end{equation}
for some $t_1$ and $t_2$ between 0 and $q-1$. If $t_2+\tau_2$ is between 0 and $q-1$, then the first coordinates of the two order pairs in~\eqref{eq:common} are equal to $h t_2+\tau_2$, with $t_2$ satisfying
\begin{equation}
 g (t_2+\tau_2) \equiv h t_2 + \tau_1 \bmod p.  \label{eq:common1}
\end{equation}
If $t_2 + \tau_2$ is larger than $q$, then $t_1 = t_2+\tau_2 - q$, and the first coordinates of the two order pairs in~\eqref{eq:common} are equal to $h_2+\tau_2$, with $t_2$ satisfying
\begin{equation}
 g (t_2+\tau_2 - q) \equiv h t_2 + \tau_1 \bmod p. \label{eq:common2}
\end{equation}
Hence,  $h t_2 + \tau_1$ may assume only two values mod $p$, one from~\eqref{eq:common1} and the other one from~\eqref{eq:common2}. Let $\mat{A}_g$ be the $p \times q$ array with characteristic set $\set{I}_g$.
We see that the elements in
$\set{I}_g \cap (\set{I}_h+(\tau_1, \tau_2))$ are located in at most two rows in $\mat{A}_g$.
Since $p$ consecutive entries in a row of $\mat{A}_g$ contain exactly one ``1'', at most two elements in $\set{I}_g \cap (\set{I}_h+(\tau_1, \tau_2))$ are covered by $\set{S}'$.

(ii) Let $\mat{A}_g$ be the $p \times q$ array with $\set{I}_{g}$ as the characteristic set. From the remark before Definition~\ref{def:CRT}, any $p$ consecutive columns in $\mat{A}_g$ form a permutation matrix. For each $b= 0,1,\ldots, p-1$, the time indices
\[
 b,\ b+p, \ldots, b+(p-1)p
\]
are $p$ consecutive entries in a row in $\mat{M}_{p,q}$. Hence there is exactly one ``1'' among $s_g(b)$, $s_g(b+p),\ldots, s_g(b+(p-1)p)$. Since this is true for $b=0,1,\ldots, p-1$, we conclude that there are exactly $p$ ``1'' in $s_g(0)$, $s_g(1), \ldots, s_g(p^2-1)$.
\end{proof}

\begin{prop} For $y = 0,1,\ldots, q-p$, let $\set{P}_y'$ be
the index set
\[
 \set{P}'_y := \{(i,j) \in G_{p,q}:\, 0 \leq i \leq p-1, y \leq j \leq y+p-1 \}.
\]
Let $\set{P}_y$ be the corresponding set of  time indices in $\mathbb{Z}_{pq}$ under the modified CRT correspondence,
\[
\set{P}_y := \Phi_{p,q}'^{-1}(\set{P}').
\]

Then

(i) For $g \neq h$, and delay offset $(\tau_1, \tau_2)$,
\[
|\set{I}_g \cap (\set{I}_h+(\tau_1,\tau_2)) \cap \set{P}'_y| \leq 2.
\]

(ii) For $g=1,2,\ldots, p$, there are exactly $p$ ones in $s_g(t)$ for $t \in \set{P}_y$.
\label{prop:P}
\end{prop}

\begin{proof}
The elements in $\set{P}'_y$ can be regarded as a square submatrix in a $p\times q$ matrix. The proof of part (i) of Prop.~\ref{prop:P} is similar to that of Prop.~\ref{prop:aperiodic} and is omitted. The second part follows from the fact that every $p$ consecutive columns in $\mat{A}_g$, which stands for the $p\times q$ matrix with characteristic set $\set{I}_g$, form a permutation matrix, and hence contains exactly $p$ ones.
\end{proof}

\begin{prop}
Let $\set{S}'$, $\set{P}'_y$  be defined as in Prop.~\ref{prop:aperiodic} and~\ref{prop:P}. For $g=1,2,\ldots, p-1$, and $\tau$ between 1 and $L-1$, the subset $\set{B}_{g,\tau}$ of time indices $t$ in $\{0,1,2,\ldots, L-1\}$,  which satisfy
\begin{align*}
s_g(t) = 1 & \text{ for } t=0,1,\ldots, \tau-1, \\
s_g(t) = 1 \text{ and } s_g(t-\tau) = 0& \text{ for } t=\tau, \tau+1,\ldots, L-1,
\end{align*}
contains $\Phi_{p,q}'^{-1}(\set{S}'\cap \set{I}_g)$, or $ \Phi_{p,q}'^{-1}(\set{P}_y'\cap \set{I}_g)$ for some~$y$.
\label{prop:forbidden}
\end{prop}

\begin{proof}
Let $\hat{s}_g(t)$ be the acyclic shift of $s_g(t)$ to the right by delay offset $\tau$. The time indices in $\set{B}_{g,\tau}$ correspond to the ``1'' in $s_g(t)$ which is not covered by $\hat{s}_g(t)$.
We consider two cases: (i) $p^2 \leq \tau <L$, and (ii) $1 \leq \tau < p^2$.

Case (i). When $p^2 \leq \tau < L$, the first $p^2$ bits in $s_g(t)$ are not covered by  $\hat{s}_g(t)$, and therefore $\set{B}_{g,\tau}$ contains $\Phi_{p,q}'^{-1}(\set{S}' \cap \set{I}_g)$.

Case (ii). Recall that in the proof of Theorem~\ref{thm:auto-correlation}, it is mentioned that the intersection of
$\set{I}_g$ and $\set{I}_g + (\tau_1, \tau_2)$ is either empty, or an arithmetic progressions in $G_{p,q}$ with common difference $(g,1)$. The intersection is non-empty if and only if $(\tau_1,\tau_2)$ equals $(g,1)k$ for some $k=\pm1, \pm2, \ldots, \pm(q-1)$.

Suppose that $\tau$ is a nonzero multiple of $p$ between 1 and $p^2-1$. By part (i) of Prop.~\ref{prop:M}, $\Phi'_{p,q}(\tau)=(0, \tau/p)$. ($\tau$ is one of the $p$ left-most entries in the first row of $\set{M}_{p,q}$.) In this case, $\Phi'(\tau)$ does not equal $(g,1)k$ for any $k \in \{1, 2, \ldots, p-1\}$. (We have used the assumption that $g$ is non-zero in this step.)  This implies that the ones in $\mat{A}_g$ with indices in
$\set{P}'_{\tau/p}$ are not covered by $\hat{s}_g(t)$. Thus, $$\set{B}_{g,\tau} \supseteq \Phi_{p,q}'^{-1}(\set{P}'_{\tau/p} \cap \set{I}_g).$$

Now suppose that $\tau$ is between 1 and $p^2-1$ but is not a multiple of~$p$. Let $(\tau_1, \tau_2) = \Phi'_{p,q}(\tau)$. Using similar argument as in the previous paragraph, if $(\tau_1, \tau_2) \not\in \set{I}_g$, then $$\set{B}_{g,\tau} \supseteq \Phi_{p,q}'^{-1}(\set{P}'_{\tau_2} \cap \set{I}_g).$$
Otherwise if $(\tau_1, \tau_2) \not\in \set{I}_g$, then the intersection of $\set{I}_g$ and $\set{I}_g + (\tau_1, \tau_2)$ equals
\[
\{ (g,1)k:\, k = \tau_2, \tau_2+1,\ldots, q-1\}.
\]
Consider the $p$ columns in $\mat{M}_{p,q}$ to the left of $\tau$, namely the time indices associated with $\set{P}'_{\tau_2-p}$ in $\mat{M}_{p,q}$. The corresponding time slots are not covered by $\hat{s}_g(t)$. Therefore,
$$\set{B}_{g,\tau} \supseteq \Phi_{p,q}'^{-1}(\set{P}'_{\tau_2-p} \cap \set{I}_g).$$
\end{proof}

\begin{figure}
\begin{center}
  \includegraphics[width=2.7in]{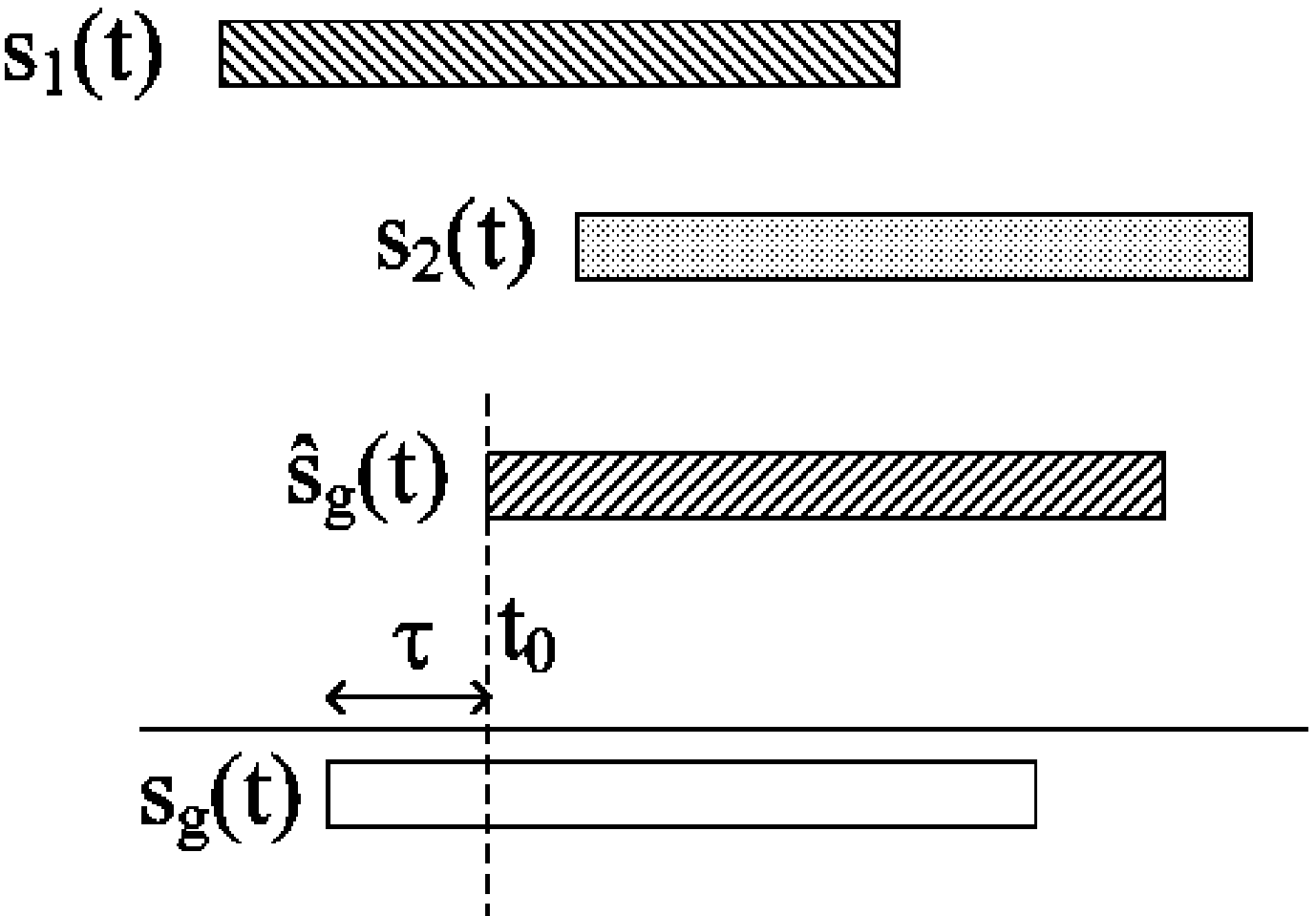}
\end{center}
\caption{The protocol sequence $\hat{s}_g(t)$ is transmitted at time $t_0$. Sequences $s_1(t)$ and $s_2(t)$ are transmitted at some other time. The receiver tries to match $c(t)$ with $s_g(t)$ at $t_0-\tau$.}
\label{fig:match}
\end{figure}

\medskip

Suppose that $\hat{s}_g(t)$ is actually transmitted at time $t_0$ and the receiver tries to match the channel-activity signal $c(t)$ with $s_g(t)$ at time $t_0-\tau$  (See Fig.~\ref{fig:match}).
If the channel-activity signal were mistakenly matched to $s_g(t)$ at $t_0 - \tau$, then by the previous proposition, either (i) the $p$ ``1'' in the first $p^2$ bits of $s_g(t)$, which occur at time indices
\[
( \Phi_{p,q}'^{-1}(\set{I}_g) \cap \{0,1,\ldots,p^2-1\}) + t - \tau,
\]
or (ii) the  $p$ ``1'' in $s_g(t)$, which occur at time indices in
\[
\Phi_{p,q}'^{-1}(\set{I}_g \cap \set{P}_y') + t - \tau,
\]
for some $y$ between 0 and $q-p$,
are covered by the other active users.
By Prop.~\ref{prop:aperiodic} and~\eqref{prop:P}, each of the other active users can contribute at most two overlapping slots. As
there are no more than $(p-1)/2$ other active users, the total number of ones that can be covered by other active users are at most $2\cdot(p-1)/2$, which is strictly smaller than~$p$. This proves that the channel-activity signal cannot be matched to $s_g(t)$ at $t_0 - \tau$ for any $\tau=1,2,\ldots, L-1$. This completes the proof of Theorem~\ref{thm:synch}.




\end{document}